%% file: ms.tex
\documentclass[letterpaper,11pt]{article}
\usepackage[utf8]{inputenc}

\pdfoutput=1

\usepackage[table,xcdraw]{xcolor}
\usepackage{amsmath, amsthm, amssymb}
\usepackage{algpseudocode,algorithm,algorithmicx}
\usepackage{algorithmicx}
\usepackage{mathtools}
\usepackage[numbers]{natbib}
\usepackage{comment} 
\usepackage{color-edits} 
\usepackage[most]{tcolorbox}
\usepackage{xfrac}
\usepackage{hyperref}
\usepackage{multirow}
\usepackage{caption}
\usepackage{bm}
\usepackage{newfloat}
\usepackage{enumitem}
\usepackage{xfrac}
\usepackage{bbm}
\usepackage{soul}
\usepackage{makecell}
\usepackage{wrapfig}
\usepackage{hhline}

\usepackage[margin=1in]{geometry}

\allowdisplaybreaks

\title{Stochastic Matching with Few Queries: New Algorithms and Tools\footnote{Portion of the work was completed while some of the authors were at Simons Institute for Theory of Computing.}}

\author{
Soheil Behnezhad\thanks{Department of Computer Science, University of Maryland. Emails: \texttt{\{soheil, farhadi, hajiagha\}@cs.umd.edu}. Supported in part by NSF CAREER award CCF-1053605, NSF AF:Medium grant CCF-1161365, NSF BIGDATA grant IIS-1546108, NSF SPX grant CCF-1822738, UMD AI in Business and Society Seed Grant, and UMD Year of Data Science Program Grant. Soheil Behnezhad was an intern at Upwork for part of the work.} \and
\hspace{-0.3cm} Alireza Farhadi\footnotemark[2] \and
\hspace{-0.3cm} MohammadTaghi Hajiaghayi\footnotemark[2] \and
\hspace{-0.3cm} Nima Reyhani\thanks{Airbnb. Email: \texttt{nima.reyhani@airbnb.com}. Portion of the work was completed while the author was at Upwork.}
}

\date{}

\input{macro.tex}
\input{utils.tex}

\begin{document}
\maketitle

\input{abstract.tex}

\clearpage

\input{intro.tex}

\input{model.tex}

\input{results.tex}

\input{nonadaptive}

\bibliographystyle{alpha}
\bibliography{refs}

\clearpage

\appendix
\input{appendix-proofs}

\input{usefulbounds}
	
\end{document}

%% file: macro.tex
\newcommand{\E}[0]{\ensuremath{\mathbb{E}}}

\newcommand{\matching}[1]{\ensuremath{M(#1)}}
\newcommand{\expmatching}[1]{\ensuremath{\mathbb{M}(#1)}}

\newcommand{\opt}[0]{\ensuremath{\textsc{opt}}}
\newcommand{\alg}[0]{\ensuremath{\textsc{alg}}}

\newcommand{\qw}[0]{\ensuremath{\varphi}}

\newcommand{\bu}{0.501}
\newcommand{\bube}{0.5011}
\newcommand{\bdbe}{0.4989}

%% file: utils.tex
\DeclareMathOperator*{\argmax}{arg\,max}

\newcommand{\Ot}[1]{\ensuremath{\widetilde{O}(#1)}}

\addauthor{sb}{blue}    
\addauthor{md}{red}  
\addauthor{ss}{purple} 

\newtheorem{result}{Result}
\newtheorem{theorem}{Theorem}[section]
\newtheorem{lemma}[theorem]{Lemma}

\newtheorem{corollary}[theorem]{Corollary}

\newtheorem{definition}[theorem]{Definition}
\newtheorem{claim}[theorem]{Claim}

\newtheorem{observation}[theorem]{Observation}
\newtheorem{remark}[theorem]{Remark}

\definecolor{mygreen}{RGB}{20,125,20}
\definecolor{mylightgray}{RGB}{230,230,230}

\hypersetup{
     colorlinks=true,
     citecolor= mygreen,
     linkcolor= mygreen
}

\algnewcommand{\IIf}[1]{\State\algorithmicif\ #1\ \algorithmicthen}
\algnewcommand{\EndIIf}{\unskip\ \algorithmicend\ \algorithmicif}

\newcounter{proc}

\newenvironment{tbox}{
\vspace{0.2cm}
\begin{tcolorbox}[width=\textwidth,
                  enhanced,
                  boxsep=2pt,
                  left=1pt,
                  right=1pt,
                  top=4pt,
                  boxrule=1pt,
                  arc=0pt,
                  colback=white,
                  colframe=black,
                  breakable]
}{
\end{tcolorbox}
}

\newcommand{\tboxhrule}[0]{\vspace{0.1cm} \hrule \vspace{0.2cm}}

\newenvironment{titledtbox}[1]{\begin{tbox}#1 \tboxhrule}{\end{tbox}}

\newenvironment{procedure}[1]{\refstepcounter{proc}\begin{titledtbox}{\textbf{Procedure \theproc.} #1}}{\end{titledtbox}}

\newenvironment{procinput}{
\textbf{Input:}
}{
\vspace{0.2cm}
}

\newcommand{\thickhline}[0]{\Xhline{1.3pt}}
\newcolumntype{?}{!{\vrule width 1.3pt}}

\newcommand{\smparagraph}[1]{\vspace{-0.3cm}\paragraph{#1}}
\newcommand{\argparagraph}[2]{\vspace{0.3cm}\noindent\textbf{#1} (#2)\textbf{.}\hspace{0.2cm}}

\newenvironment{graytbox}{
\vspace{0.1cm}
\begin{tcolorbox}[width=\textwidth,
                  enhanced,
                  frame hidden,
                  boxsep=3pt,
                  left=1pt,
                  right=1pt,
                  top=1pt,
                  boxrule=1pt,
                  arc=0pt,
                  colback=mylightgray,
                  colframe=black,
                  breakable
                  ]
}{
\end{tcolorbox}
}

%% file: abstract.tex
\begin{abstract}
\setlength{\parskip}{0.4em}
	We consider the following {\em stochastic matching} problem on both weighted and unweighted graphs: A graph $G=(V, E)$ along with a parameter $p\in (0, 1)$ is given in the input. Each edge of $G$ is {\em realized} independently with probability $p$. The  goal is to select a degree bounded (dependent only on $p$) subgraph $H$ of $G$ such that the expected size/weight of maximum realized matching of $H$ is close to that of $G$.
	
	This model of stochastic matching has attracted significant attention over the recent years due to its various applications in kidney exchange, online labor markets, and other matching markets. The most fundamental open question is the best approximation factor achievable for such algorithms that, in the literature, are referred to as {\em non-adaptive} algorithms. Prior work has identified breaking (near) half-approximation as a barrier for both weighted and unweighted graphs. Our main results are as follows:
	
	\begin{itemize}
	
	\item We analyze a simple and clean algorithm and show that for unweighted graphs, it finds an (almost)  $4\sqrt{2}-5$ ($\approx 0.6568$) approximation by querying $O(\frac{\log (1/p)}{p})$ edges per vertex. This improves over the  state-of-the-art $0.5001$ approximation of Assadi et al.~[EC'17].
	
	\item We show that the same algorithm achieves a $0.501$ approximation for weighted graphs by querying $O(\frac{\log (1/p)}{p})$ edges per vertex. This is the first algorithm to break $0.5$ approximation barrier for weighted graphs. It also improves the per-vertex queries of the state-of-the-art by Yamaguchi and Maehara~[SODA'18] and Behnezhad and Reyhani~[EC'18].
	\end{itemize}
	
	Prior results were all interestingly based on similar algorithms and differed only in the analysis. Our algorithms are fundamentally different, yet very simple and natural. For the analysis, we introduce a number of procedures that construct heavy fractional matchings. We consider the new algorithms and our analytical tools to be the main contributions of this paper.
\end{abstract}


%% file: intro.tex
\section{Introduction}\label{sec:intro}

We consider the following {\em stochastic matching} problem on both weighted and unweighted graphs. In its most general form, an edge-weighted graph $G=(V, E, w)$ along with a parameter $p \in (0, 1)$ is given in the input and each edge of $G$ is {\em realized} independently with probability $p$. We are unaware of the edge realizations yet our goal is to find a heavy realized matching. To do this, we can select a degree-bounded (i.e., dependent only on $p$) subgraph $Q$ of $G$, {\em query} all of its edges simultaneously, and report its maximum realized matching.  Denoting the expected weight of the maximum realized matching of any subgraph $H$ of $G$ by $\expmatching{H}$, the goal is choose $Q$ such that it maximizes $\expmatching{Q}/\expmatching{G}$ --- which is also known as the approximation factor.


The restriction on the number of queries per vertex comes from the fact that the querying process is often {\em time consuming} and/or {\em expensive} in the applications of stochastic matching. Without this restriction, the  solution is trivial as one can simply query all the edges of $G$ and report the maximum matching among those that are realized. 

The algorithms in this setting are categorized as {\em non-adaptive} since they query all the edges simultaneously without any prior knowledge about the realizations. In contrast, {\em adaptive} algorithms have multiple rounds of adaptivity and the queries conducted at each round can depend on the outcome of the prior queries. Non-adaptive algorithms are considered practically more desirable since the queries are not stalled behind each other. In fact, one can see a non-adaptive algorithm as  an adaptive algorithm that is restricted to have only one {\em round} of adaptivity; therefore, it is not hard to see that it is generally much more complicated to design and analyze non-adaptive algorithms.

While $(1-\epsilon)$-approximate adaptive algorithms are known, even for weighted graphs, the literature has identified breaking half approximation to be a barrier for non-adaptive algorithms \cite{DBLP:conf/sigecom/BlumDHPSS15, DBLP:conf/sigecom/AssadiKL16, DBLP:conf/sigecom/AssadiKL17, DBLP:conf/soda/YamaguchiM18, DBLP:conf/sigecom/BehnezhadR18}. Prior to our work, no such algorithm was known for weighted graphs and even for unweighted graphs, the state-of-the-art non-adaptive algorithm of Assadi et al.~\cite{DBLP:conf/sigecom/AssadiKL17} achieves only a slightly better approximation factor of $0.5001$.


We introduce new algorithms and techniques to bypass these bounds. For unweighted graphs, we achieve a 0.6568 approximation and show that the same algorithm bypasses 0.5 approximation for weighted graphs. In both algorithms, we query only $O(\log(1/p)/p)$ edges per-vertex. These results answer several open questions of the literature that we elaborate more on in the forthcoming paragraphs. Apart from the approximation factor, it is not hard to see that any algorithm achieving a constant approximation has to query $\Omega(1/p)$ edges per vertex (see e.g., \cite{DBLP:conf/sigecom/AssadiKL16}). As such, the number of per-vertex queries conducted by our algorithms is optimal up to a factor of $O(\log 1/p)$.


\input{relatedwork}

\input{applications}

\input{furtherrelatedwork}

%% file: relatedwork.tex

\paragraph{Prior work.} The stochastic matching problem has been intensively studied during the past decade due to its diverse applications from {\em kidney exchange} to {\em labor markets} and {\em online dating} (we overview these applications in Section~\ref{sec:applications}). Directly related to the setting that we consider are the papers by Blum et al.~\cite{DBLP:conf/sigecom/BlumDHPSS15} (which introduced this variant of stochastic matching), Assadi et al.~\cite{DBLP:conf/sigecom/AssadiKL16,DBLP:conf/sigecom/AssadiKL17}, Yamaguchi and Maehara \cite{DBLP:conf/soda/YamaguchiM18}, and Behnezhad and Reyhani \cite{DBLP:conf/sigecom/BehnezhadR18}. Table~\ref{table:relatedwork} gives a brief survey of known results due to these papers as well as a comparison to our results. We give a more detailed description of the main differences below.


\begin{table}
\centering
\resizebox{9.5cm}{!}{%
\begin{tabular}{?l|l|c|c?} 
\thickhline
                            & Reference  & \multicolumn{1}{l|}{Apx factor}                                                                                                                              & \multicolumn{1}{l?}{Per-vertex queries}                    \\ 
\hline
\multirow{4}{*}{Unweighted} &      \cite{DBLP:conf/sigecom/BlumDHPSS15}      & $0.5-\epsilon$                                                                                                                                               & $\widetilde{O}(1/p^{2/\epsilon})$                          \\ 
\cline{2-4}
                            &       \cite{DBLP:conf/sigecom/AssadiKL16}     & $0.5-\epsilon$                                                                                                                                               & $\widetilde{O}(1/\epsilon p)$                              \\ 
\cline{2-4}
                            &       \cite{DBLP:conf/sigecom/AssadiKL17}     & $0.5001$                                                                                                                                                     & $\widetilde{O}(1/p)$                                       \\ 
\hhline{|~|-|-|-|}
                            & \textbf{This paper} & \begin{tabular}[c]{@{}>{\cellcolor[rgb]{1,1,1}}c@{}}$0.6568$\\ $\big(\approx 4\sqrt{2}-5\big)$ \end{tabular} & $\widetilde{O}(1/p)$   \\ 
\thickhline
\multirow{4}{*}{Weighted}   &      \cite{DBLP:conf/soda/YamaguchiM18}      & \multirow{3}{*}{$0.5-\epsilon$ }                                                                                                                             & $\widetilde{O}(W\log(n)/\epsilon p)$                       \\ 
\cline{2-2}\cline{4-4}
                            & \cite{DBLP:conf/soda/YamaguchiM18} (B)        &                                                                                                                                                              & $\widetilde{O}(W/\epsilon p)$                              \\ 
\cline{2-2}\cline{4-4}
                            &       \cite{DBLP:conf/sigecom/BehnezhadR18}     &                                                                                                                                                              & $\widetilde{O}(1/\epsilon p^{4/\epsilon})$                 \\ 
\hhline{|~|-|-|-|}
                            & \textbf{This paper} & $0.501$                                                                                                                  & $\widetilde{O}(1/p)$   \\
\thickhline
\end{tabular}
}%
\caption{Bounds known for non-adaptive algorithms. We have hidden $\log(1/\epsilon p)$ factors to simplify comparison. The result indicated with (B) in the reference assumes that the input graph is bipartite.}
\label{table:relatedwork}
\vspace{-0.1cm}
\end{table}

Blum et al. introduced the following algorithm:

\begin{tbox}
	\textbf{\hypertarget{alga}{Algorithm}} $\mathcal{A}$ (\cite{DBLP:conf/sigecom/BlumDHPSS15}): {\em Pick a maximum matching $M_i$ from $G$ and remove all of its edges. Repeat this for $R$ iterations, then query the edges in $M_1 \cup \ldots \cup M_R$ simultaneously and report the maximum realized matching among them.}
\end{tbox}

\newcommand{\alga}[0]{Algorithm~\hyperlink{alga}{$\mathcal{A}$}}
\newcommand{\algant}[0]{Algorithm~$\mathcal{A}$}

It is easy to see that $R$, in \alga{}, determines the per-vertex queries. This means that it suffices to argue that a small value for $R$ is sufficient to get our desired approximation factors. Blum et al.~\cite{DBLP:conf/sigecom/BlumDHPSS15} showed that for unweighted graphs, setting $R = 1/p^{O(1/\epsilon)}$ is sufficient to get a $0.5-\epsilon$ approximation. Interestingly, the follow-up results were achieved by the same algorithms (with minor changes) and differed mainly in the analysis. Assadi et al.~\cite{DBLP:conf/sigecom/AssadiKL16} showed that setting $R= \widetilde{O}(1/\epsilon p)$ suffices to achieve a $0.5-\epsilon$ approximation improving the exponential dependence on $1/\epsilon$.\footnote{The algorithm of Assadi et al.~\cite{DBLP:conf/sigecom/AssadiKL16} also incorporates a {\em sparsification step} to ensure $\opt = \Omega(n)$.} Yamaguchi and Maehara~\cite{DBLP:conf/soda/YamaguchiM18} generalized these results to weighted graphs.\footnote{The generalization of Blum et al.'s algorithm to weighted graphs is simply to pick maximum weighted matchings in each round/iteration.} They showed that it suffices to set $R=O(W \log n /\epsilon p)$ to achieve the same approximation factor of $0.5-\epsilon$ where $W$ denotes the maximum integer edge weight. Behnezhad and Reyhani \cite{DBLP:conf/sigecom/BehnezhadR18} further showed that the same approximation factor of $0.5-\epsilon$ can be achieved for weighted graphs by setting $R = O(1/\epsilon p^{4/\epsilon})$. While this removes the dependence on $W$ and $n$, making the bound a constant, it has a worse dependence on $1/\epsilon$ than that of \cite{DBLP:conf/soda/YamaguchiM18}.

Observe that the approximation factor of all the algorithms mentioned above is the same. The only exception in the literature is the algorithm of Assadi et al.~\cite{DBLP:conf/sigecom/AssadiKL17} which achieves a $0.5001$ approximation for unweighted graphs. Their algorithm first extracts a {\em large} $b$-matching (which depends on the expected size of the realized matching) from the graph and then applies \alga{} on the remaining graph. They interestingly show that the edges chosen by \alga{} can be used to augment the realized matching among the edges of the $b$-matching which leads to bypassing the half approximation barrier for unweighted graphs.

\smparagraph{Our contribution.} Despite the theoretical guarantees of the literature for \alga{}, it has its drawbacks. Blum et al.~\cite[Theorem 5.2]{DBLP:journals/corr/BlumHPS14} give examples on which it does not achieve better than a $\sfrac{5}{6}$ approximation. It also seems notoriously difficult (if not impossible) to analyze anything better than a $0.5$ approximation for \alga{} alone. We consider another algorithm which is also very simple and natural:
\begin{tbox}
	\textbf{Algorithm} \hypertarget{algb}{$\mathcal{B}$} (Formally as Algorithm~\ref{alg:nonadaptive}): First draw $R$ realizations $\mathcal{G}_1, \ldots, \mathcal{G}_R$ of $G$ independently. Then from each of these realizations $\mathcal{G}_i$, pick a maximum (weighted) matching $M_i$. Finally, query the edges that appear in $M_1 \cup \ldots \cup M_R$ simultaneously and report the maximum realized matching among them.
\end{tbox}

\newcommand{\algb}[0]{Algorithm~\hyperlink{algb}{$\mathcal{B}$}}
\newcommand{\algbnt}[0]{Algorithm~$\mathcal{B}$}

Similar to \alga{}, here $R$ determines the number of per-vertex queries. We analyze \algb{} for both weighted and unweighted graphs.

\begin{graytbox}
\begin{result}[formally as Theorem~\ref{thm:nonadaptiveweighted}]\label{res:weightednonadaptive}
	For $R=O(\frac{\log (1/p)}{p})$, \algb{} achieves a $0.501$ approximation on weighted graphs.
\end{result}
\end{graytbox}

Result~\ref{res:weightednonadaptive} implies the first non-adaptive algorithm that breaks the $0.5$ approximation barrier for weighted graphs. The number of per-vertex queries of this result also improves that of $0.5-\epsilon$ approximations of \cite{DBLP:conf/soda/YamaguchiM18} and \cite{DBLP:conf/sigecom/BehnezhadR18}. 


\begin{graytbox}
\begin{result}[formally as Theorem~\ref{thm:nonadaptiveunweighted}]\label{res:unweighted}
	For $R=O(\frac{\log (1/p)}{p})$, \algb{} achieves a \mbox{$0.6568$} approximation on unweighted graphs.
\end{result}
\end{graytbox}

Result~\ref{res:unweighted} improves over the state-of-the-art $0.5001$ approximate algorithm of Assadi et al.~\cite{DBLP:conf/sigecom/AssadiKL17}.\footnote{For the case of unweighted graphs, in an independent work, Assadi and Bernstein~\cite{assadisosa} give an (almost) $2/3$ approximation which is slightly better than our factor. Their algorithm, however, is highly tailored for unweighted graphs and gives no guarantee for the weighted case.}

%
%

In our analysis, we devise different {\em procedures}, that given query outcomes, they construct large {\em fractional matchings} over the realized edges. Then based on the size of this fractional matching, we get that there must also be a large integral realized matching. We give more high-level ideas and intuitions about these procedures in Section~\ref{sec:techniques}.

%% file: applications.tex
\subsection{Applications}\label{sec:applications}

The stochastic matching problem has a wide range of applications from {\em kidney exchange} to {\em labor markets} and {\em online dating}. In all these applications, the goal is to find a large (or heavy) matching and the main bottleneck is determining which edges exist in the graph. We overview some of these applications below.

\smparagraph{Kidney exchange.} Transplant of a kidney from a living {\em donor} is possible if the recipient ({\em patient}) happens to be medically compatible with his/her  donor. This is not always the case, however, kidney exchange provides a way to overcome this. In its simplest form with {\em pairwise exchanges}, two incompatible donor/patient pairs can exchange kidneys. That is, the donor of the first pair donates kidney to the patient of the second pair and vice versa. This gives rise to the notion of a {\em compatibility graph} where we have one vertex for each incompatible donor/patient pair and each edge determines the possibility of an exchange. Therefore, the pairwise exchanges that take place can be expressed as a matching of this graph. There is, however, one crucial problem. The medical records of the patients such as their blood- or tissue-types only rule out a subset of incompatibilities. For the rest, we need more accurate medical tests that are both costly and time consuming.

The stochastic matching setting helps in finding a large matching among the pairs who also pass the extra tests while conducting very few medical tests per pair. There is a rich literature on such algorithmic approaches for kidney exchange particularly in stochastic settings \cite{DBLP:conf/sigecom/AkbarpourLG14, DBLP:conf/soda/AndersonAGK15, anderson2015finding, DBLP:conf/ijcai/AwasthiS09, DBLP:conf/aaai/DickersonPS12, DBLP:conf/sigecom/DickersonPS13, DBLP:conf/aaai/DickersonS15, DBLP:journals/jea/ManloveO14, unver2010dynamic}. We refer interested readers to the paper of \cite{DBLP:conf/sigecom/BlumDHPSS15} for a more detailed discussion about the application of stochastic matching in kidney exchange.

\smparagraph{Online labor markets.} Online labor markets facilitate working relationships between freelancers and employers. In such platforms, it quite often happens that the users (from either party) have more options than they can consider. We can represent this with a bipartite graph with freelancers on one side and employers on the other. The edges of the {\em compatibility graph}, again, determine possible matches. While the initial job descriptions rule out some of the edges, it is after an interview between an employer and the freelancer that they decide whether to work with each other. Stochastic matching, for such platforms, can be used to recommend interviews. This way, we ensure that with very few interviews, most of the users will find a desired match.



%% file: furtherrelatedwork.tex
\smparagraph{Further related work.}

Multiple variants of stochastic matching have been considered by prior work. A well-studied setting, first introduced by Chen et al.~\cite{DBLP:conf/icalp/ChenIKMR09}, is the {\em query-commit} model. In this model, the queried edges that happen to be realized have to be included in the final matching \cite{DBLP:journals/ipl/Adamczyk11, DBLP:journals/algorithmica/BansalGLMNR12, DBLP:conf/icalp/ChenIKMR09, DBLP:conf/icalp/CostelloTT12, DBLP:conf/ipco/GuptaN13}. Another related setting is the model of \cite{DBLP:conf/sigecom/BlumGPS13} which allows to query only two edges per vertex. We refer to \cite{DBLP:conf/sigecom/BlumDHPSS15} for a more extensive overview of other models relevant to the one we consider.

%% file: model.tex
\section{Preliminaries}\label{sec:preliminaries}
\paragraph{Notation.} For any edge set $E$, we denote by $\matching{E}$ the weight of the maximum weighted matching in $E$. We may also abuse notation throughout the paper and use $\matching{E}$ to refer to the set of edges in the maximum weighted matching of $E$. When it is clear from the context, we may use maximum matching instead of maximum {\em weighted} matching. For any $U \subseteq V$, we use $G[U]$ to denote the induced subgraph of $G$ over $U$.

\subsection{The Model of Stochastic Matching}

We are given a graph $G=(V, E)$ with edge weights $w: E \to \mathbb{R}_+$ along with a fixed parameter $p\in [0, 1]$. Each of the edges in $E$ is {\em realized} independently from other edges with probability $p$. The {\em realized graph} $G_p=(V, E_p)$ includes an edge $e \in E$ if and only if it is realized. We are not initially aware of the realized graph $G_p$. Our goal, however, is to compute a heavy matching of $G_p$. To do so, we can {\em query} each edge in $E$ and the outcome is whether the edge is realized.

For any $E' \subseteq E$, we denote by $\expmatching{E'} := \E[\matching{E' \cap E_p}]$ the expected weight of the realized matching in $E'$. The benchmark in the stochastic matching problem is the omniscient optimum matching $\expmatching{E}$, which we also denote by $\opt$.  A {\em non-adaptive} algorithm in this setting, has to pick a degree-bounded (dependent only on $1/p$) subgraph $Q$ of $G$ such that $\expmatching{Q}/\opt$, which determines the approximation factor, is maximized. If the algorithm is randomized, which is the case in our paper, it should succeed with high probability.\footnote{We note that throughout the paper, for simplicity, we analyze the approximation factor of our algorithms in expectation. However, it is easy to boost the success probability to $1-o(1)$ by running several instances of the algorithm to obtain candidate solutions $Q_1, \ldots, Q_k$, and then reporting $Q := \argmax_{Q_i} \expmatching{Q_i}$ as the solution.}

%

\subsection{Background on the Matching Polytope}
Fix a graph $G=(V, E)$. A vector $x \in \mathbb{R}^{E}$ is a {\em fractional matching} of $G$ if for any $e \in E$, we have $x_e \geq 0$ and for any $v \in V$ we have $x_v := \sum_{e \ni v} x_e \leq 1$. An integral matching can be seen as a fractional matching where for any $e \in E$ we have $x_e \in \{0, 1\}$. The {\em matching polytope} $\mathcal{P}(G)$ of $G$, is the convex hull of all integral matchings of $G$ represented as above.  Edmonds~\cite{edmonds1965maximum} showed in 1965 that $\mathcal{P}(G)$ is the solution set of linear program:
\begin{flalign*}
	&& & x_e \geq 0 \qquad &\forall e \in E && \\
	&& & x_v \leq 1 \qquad &\forall v \in V && \\
	&& & x(U) \leq \lfloor |U|/2 \rfloor &\forall U \subseteq V \text{ with odd } |U| &&
\end{flalign*}
where $x(U)$ denotes $\sum_{e \in G[U]} x_e$. Note that the first two constraints only ensure that $x$ is a valid fractional matching. Constraints of the third type are known as {\em blossom inequalities}. A corollary of Edmond's theorem is the following:

\begin{corollary}\label{cor:edmonds}
	Let $x$ be a fractional matching of an edge weighted graph $G$ that satisfies blossom inequalities, i.e., $x \in \mathcal{P}(G)$. Then $G$ has an integral matching $y$ where $\sum_e y_e w_e \geq \sum_e x_e w_e$.
\end{corollary}

We can even relax the blossom inequalities and consider only subsets of size at most $1/\epsilon$, and ensure that the weight of no fractional matching exceeds maximum weight of integral matchings by a larger than $1/(1-\epsilon)$ factor. This is captured by the following folklore lemma.

\begin{lemma}[folklore]\label{lem:folklore}
	Let $x$ be a fractional matching of an edge weighted graph $G$ where for any $U \subseteq V$ with $|U| \leq 1/\epsilon$, it satisfies $x(U) \leq \lfloor |U|/2 \rfloor$. Then $G$ has an integral matching $y$ where $\sum_e y_e w_e \geq (1-\epsilon) \sum_e x_e w_e$.
\end{lemma}
\begin{proof}[Proof sketch]
	Define $z = x/(1+\epsilon)$. Since $x_v \leq 1$ for any $v$, one can show easily that $z$ satisfies all blossom inequalities. Therefore, by Corollary~\ref{cor:edmonds}, there must exist an integral matching of weight at least that of $z$ which by definition is $\sum_e z_e w_e = (\sum_e x_e w_e)/(1+\epsilon) \geq (1-\epsilon)\sum_e x_e w_e$.
\end{proof}

We refer interested readers to Section~25.2 of \cite{schrijver2003combinatorial} for a comprehensive overview of the matching polytope.

%% file: results.tex
\section{Technical Overview}\label{sec:techniques}

To give an intuition about the true differences between our algorithm (\algb{}) and the standard non-adaptive algorithm of the literature (\alga{}), we start by restating the bad example of Blum et al.~\cite[Theorem 5.2]{DBLP:journals/corr/BlumHPS14} for \alga{} and describing how \algb{} overcomes it. We then proceed to give intuitions on how we analyze the performance of \algb{}.

\smparagraph{A comparison of \algant{} and \algbnt{}.} Consider the graph $G=(V, E)$ of Figure~\ref{fig:badexample}-(a) whose vertex set is partitioned into six subsets $A$, $B_1$, $B_2$, $C_1$, $C_2$, and $D$, each of size $N$. The edge set of the graph contains complete bipartite graphs between pairs $(A, B_1)$, $(A, B_2)$, $(D, C_1)$, and $(D, C_2)$ and perfect matchings between pairs $(B_1, C_1)$ and $(B_2, C_2)$. Assume also that the realization probability $p$ is $0.5$.

\begin{figure}[h]
  \centering
  \includegraphics[scale=0.85]{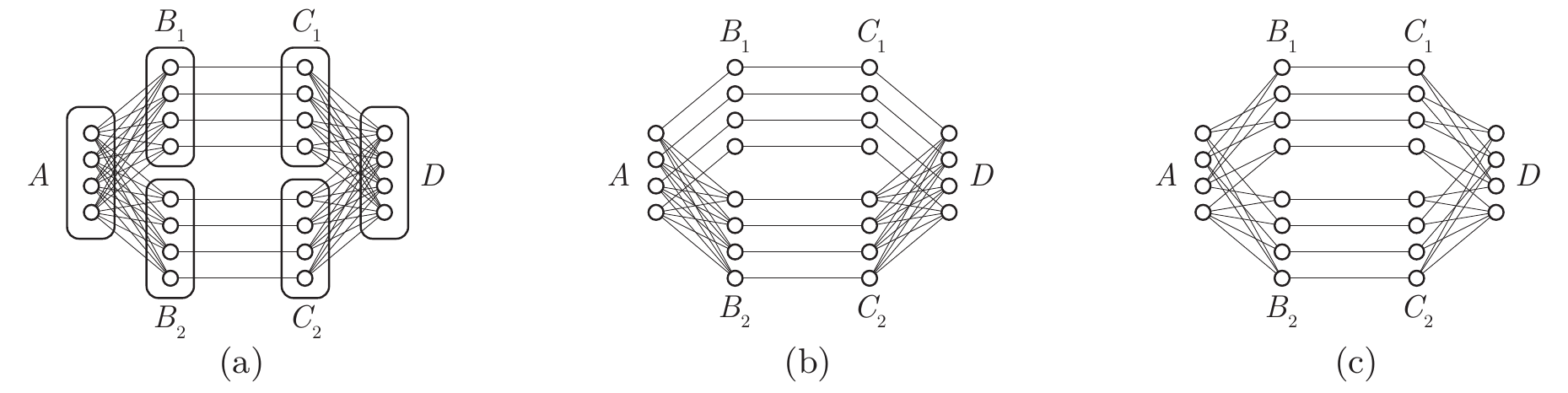}
  \vspace{-0.2cm}
  \caption{Figure (a) illustrates the input graph. Figure (b) illustrates a potential subset of queried edges by \algant{}. Figure (c) illustrates the expected structure of queried edges of \algbnt{}.}
  \label{fig:badexample}
\end{figure}
\vspace{-0.1cm}


It is not hard to confirm that the expected omniscient optimum matching of $G_p$ is an almost perfect matching of size $3N - o(N)$. It suffices to add the realized edges between $(B_1, C_1)$ and $(B_2, C_2)$ to $\opt$ which roughly matches half of the vertices of each of these sets in expectation and then find large realized matchings between the remaining vertices and those in $A$ and $D$.

Recall that \alga{} picks an arbitrary maximum matching $M_i$ in each iteration and removes it from the graph. Suppose that these matchings are as follows: The first matching $M_1$ contains the edges in $(B_1, C_1)$, a perfect matching in $(A, B_2)$, and a perfect matching in $(D, C_2)$. Matching $M_2$ contains the edges in $(B_2, C_2)$, a perfect matching in $(A, B_1)$, and a perfect matching in $(D, C_1)$. Each of the remaining matchings $M_3, \ldots, M_R$ is the union of a perfect matching in $(A, B_2)$ and a perfect matching in $(D, C_2)$. The queried edges by \alga{} are illustrated in Figure~\ref{fig:badexample}-(b). Since for every vertex in $B_1$ or $C_2$, only two edges are queried and $p=0.5$, we expect $1/4$ fraction of these vertices to have no realized queried edges. This means that \alga{} cannot construct a near perfect matching.

Since \algb{} incorporates a randomization throughout the process, particularly in choosing realizations $\mathcal{G}_1, \ldots, \mathcal{G}_R$ from which it picks matchings $M_1, \ldots, M_R$, bad cases such as the one described above cannot happen. In particular, for the graph of Figure~\ref{fig:badexample}, for every vertex in $B$ or $C$, in roughly half of the realizations, they are matched to a vertex in $A$ and $D$, thus we query $\widetilde{\Omega}(R/2)$ edges for each of these vertices and it is not hard to show that for a constant $R$ depending only on $\epsilon$ and $p$, \algb{} achieves a $1-\epsilon$ approximation for this example (see Figure~\ref{fig:badexample}-(c)).

\paragraph{Roadmap for analyzing \algbnt{}.} To convey the main intuitions behind the analysis, we make a few simplifying assumptions. First, assume that the input graph is unweighted. Denote the set of queried edges of \algb{} by $S$ and further denote by $S_p$ those edges in $S$ that are realized. Our goal is to show that in expectation, there exists a matching of size $0.65\opt$ in $S_p$, or in other words, $\expmatching{S} \geq 0.65\opt$. To do this, by Lemma~\ref{lem:folklore}, it suffices to show that there exists a fractional matching of size $0.65\opt$ in $S_p$ that also satisfies blossom inequalities. Let us further assume that $G$ is bipartite so that any fractional matching satisfies blossom inequalities automatically.

Denote by $q_e$ the probability that edge $e$ appears in the omniscient optimum matching.\footnote{We assume that given a realization, the edges that belong to the maximum matching are unique. This can be achieved by using a deterministic matching algorithm.} Recall that in each iteration of \algb{}, we draw a realization and add its maximum matching to $S$. Therefore, $q_e$ also denotes the probability that we sample edge $e$ in each iteration of \algb{}. One can easily confirm that for any vertex $v$, we have $\sum_{e \ni v}q_e \leq 1$. Therefore, one can think of $q_e$'s as a fractional matching with some other nice properties. Denote this fractional matching by $q$. The reader soon notices the following useful properties of $q$:
\begin{enumerate}[label={(P\arabic*)}]
	\item For any edge $e$, we have $q_e \leq p$.\\ {\em Proof sketch. Each edge is realized w.p.\footnote{Throughout, we use w.p. to abbreviate ``with probability".} $p$ and thus appears in \opt{} w.p. at most $p$.}
	\item For any set $F \subseteq E$, the expected matching \expmatching{F} of $F$ has size at least $q(F) := \sum_{e\in F} q_e$.\\ {\em Proof sketch. Suffices for each realization $E_p$ of $E$ to consider matching $F \cap \matching{E_p}$.}
\end{enumerate}

\begin{figure}[hbt]
  \centering
  \vspace{-0.4cm}
  \includegraphics[scale=0.9]{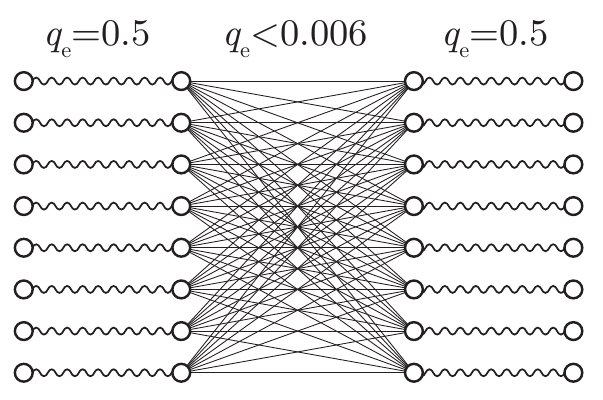}
  \vspace{-0.2cm}
  \caption{}
  \vspace{-0.3cm}
  \label{fig:crucialnoncrucial}
\end{figure}

We set a threshold $\tau \approx \delta p$ for a sufficiently small constant $\delta < 1$ and partition $E$ into two subsets of {\em crucial} edges $C := \{ e \, | \, q_e \geq \tau\}$ and {\em non-crucial} edges  $N := \{e \, | \, q_e < \tau\}$. Figure~\ref{fig:crucialnoncrucial} illustrates the values of $q_e$ over a simple example for which $p=0.5$. In this example, each wavy edge on the side that is realized appears in \opt{}, thus they all have $q_e = p = 0.5$ and  are crucial. The edges in between are significantly less likely to be in \opt{} and for all of them $q_e < 0.006$, thus they are all considered non-crucial.

Note that $q$ is merely a function of the graph's structure and is independent of our algorithms. Our goal is to show that within only $R = \Ot{1/\tau} = \Ot{1/p}$ iterations, \algb{} achieves our desired guarantee. To do this, we prove two canonical lemmas.

\argparagraph{Crucial edges lemma}{Formally as Lemma~\ref{lem:sampleallcrucial}} \algb{} samples almost all crucial edges. Therefore, by (P2), the expected matching $\expmatching{S \cap C}$ has size at least $(1-\epsilon) q(C)$ where $\epsilon$ is any desirably small constant ($\epsilon$ and $\delta$ are interdependent).

\vspace{0.2cm}

For non-crucial edges, the argument above does not work. The reason is that, as illustrated in Figure~\ref{fig:crucialnoncrucial}, the number of non-crucial edges connected to each vertex can be much more than the maximum degree of $S$ (which determines the number of per-vertex queries), thus, we can only sample a small portion of non-crucial edges which means $q(S \cap N)$ can be arbitrarily smaller than $q(N)$. Instead, we take a different approach for non-crucial edges.

\argparagraph{Non-crucial edges lemma}{Formally as Lemma~\ref{lem:non-cruciallemma}} One can construct a fractional matching $x$ over the realized non-crucial edges of $S$ (i.e., over the edges in $E_p \cap S \cap N$) whose size is at least $(1-\epsilon)q(N)$.  Moreover, for any vertex $v$, $x_v$ is no more than $\max\{q^N_v, \epsilon\}$ where we call $q^N_v := \sum_{e \ni v : e\in N} q_e$ the {\em non-crucial budget} of each vertex.

\vspace{0.2cm}

The precise proof of the non-crucial edges lemma is out of the scope of this section. However, it relies critically on the fact that $q_e$ of non-crucial edges is small. For example, if we use the same technique to construct a fractional matching for the crucial edges, we only end up with a fractional matching of size $\approx 0.4 q(C)$.

The combination of the two lemmas above immediately implies a $0.5 - \epsilon$ approximation. For this, one can easily show that $q(C) + q(N) = \opt$, and thus, either $q(C) \geq 0.5\opt$ or $q(N) \geq 0.5\opt$. For the former case, we can use the crucial edges lemma to argue that we get an almost $0.5$ approximation and for the latter we can use the non-crucial edges lemma. However, as mentioned before, our goal is to provide a much better approximation guarantee than $0.5-\epsilon$. Therefore, we have to show that the realized portions of the crucial and non-crucial edges can be augmented to construct a much larger matching. To do this, we have to  devise more involved procedures that construct large fractional matchings over the realized edges of $S$ by combining both crucial and non-crucial edges. Note that these procedures are merely analytical tools and our algorithm is still \algb{}.

For unweighted graphs, the procedure that we use --- formalized as Procedure~\ref{proc:crucial} --- is roughly as follows: We first use the non-crucial edges lemma to construct a fractional matching $x$ of size $(1-\epsilon)q(N)$ on the non-crucial edges without ``looking" at the realization of crucial edges. Independently, we reveal realized crucial edges, and pick a large realized matching $\mu^C$ among them.\footnote{For technical details,  matching $\mu^C$ is not simply the largest realized matching of crucial edges and has to be drawn according to a specific distribution. See Procedure~\ref{proc:crucial} for more details.} Then in our fractional matching $x$, we allocate the maximum possible fractional matching value to the edges in $\mu^C$ while ensuring that $x$ remains a valid fractional matching.

In Theorem~\ref{thm:nonadaptiveunweighted}, we give an analysis that shows Procedure~\ref{proc:crucial}  in expectation constructs a fractional matching of size $(1-\epsilon)(4\sqrt{2}-5)\opt$. This implies that \algb{} achieves an (almost) $(4\sqrt{2}-5)\approx 0.6568$ approximation. We note that in the analysis, the second property of non-crucial edges lemma, where we show the non-crucial budget of each vertex is not violated by the constructed fractional matching plays an important role.

While we have no upper bound on the best provable approximation factor for \algb{}, we show that at least for Procedure~\ref{proc:crucial}, our analysis is tight. That is, we give an example in Lemma~\ref{lem:unweightedtight} for which the fractional matching constructed by Procedure~\ref{proc:crucial} has size no more than $(4\sqrt{2}-5+o(1))\opt$.

\smparagraph{Generalization to weighted graphs.} In generalizing our results to weighted graphs, we follow the same approach in partitioning the edges into crucial and non-crucial subsets. In fact, both the crucial and non-crucial edges lemmas can be adapted seamlessly to the weighted graphs leading to a simple (almost) half approximation as described above. However, we show that a large class of procedures (including Procedure~\ref{proc:crucial}) achieve no more than a $0.5$ approximation for weighted graphs. The authors find this strikingly surprising which further highlights the true challenge in beating half approximation for weighted graphs. As a result, the procedure that we use to bypass half approximation for weighted graphs (formalized as Procedure~\ref{proc:crucialweighted}) is much more intricate and achieves an approximation factor of only $0.501$ (see Theorem~\ref{thm:nonadaptiveweighted}).

%% file: nonadaptive.tex
\section{The Algorithm}\label{sec:nonadaptive}

In this section, we introduce a non-adaptive algorithm formalized as Algorithm~\ref{alg:nonadaptive} as well as a number of analytical tools that we use in analyzing it for weighted and unweighted graphs. We note that for the sake of brevity, we did not attempt to optimize the constant factors in the description of Algorithm~\ref{alg:nonadaptive}.

\begin{algorithm}
  \caption{A non-adaptive algorithm for the weighted stochastic matching problem.}
  \label{alg:nonadaptive}
  \begin{algorithmic}[1]
  	\Statex \textbf{Input:} Input graph $G=(V, E)$, edge weights $w: E \to \mathbb{R}_+$ and realization probability $p \in [0, 1]$.
  	\Statex \textbf{Parameter:} $R = \frac{2000 \log(1/\epsilon)\log(1/\epsilon p)}{\epsilon^4 p}$.
  	\State $S \gets \emptyset$
	\For{$r = 1, \ldots, R$}
		\State Construct a realization $\mathcal{G}_r=(V, \mathcal{E}_r)$ of $G$, where any edge $e \in E$ appears in $\mathcal{E}_r$ independently with probability $p$.
		\State Add the edges in maximum weighted matching $\matching{\mathcal{E}_r}$ of $\mathcal{G}_r$ to to $S$.
	\EndFor
	\State Query the edges in $S$ and report the maximum weighted matching of it.
  \end{algorithmic}
\end{algorithm}

The main challenge in analyzing Algorithm~\ref{alg:nonadaptive} comes from the fact that the realizations $\mathcal{G}_1, \ldots, \mathcal{G}_R$ that are picked may be very different from the actual realization $G_p$ of $G$ on which the algorithm has to perform well. Take, for instance, the maximum matching $M_1$ of $\mathcal{G}_1$ that we add to $S$ during the first iteration of Algorithm~\ref{alg:nonadaptive}. Since the realization $\mathcal{G}_1$ is drawn from the same distribution that the actual realization $G_p$ is drawn from, one can argue that $M_1$ is as large as $M(E_p)$ in expectation. However, the problem is that only $p$ fraction of the edges in $M_1$ are expected to appear in $E_p$. This means that the realized matching $M(M_1 \cap E_p)$ found by round 1 guarantees only an approximation factor of $p$ which can be arbitrarily small. To achieve our desired approximation factor, we need to argue that the realized edges of $M_1, \ldots, M_R$ can be combined with each other to construct a heavy matching. To show this, we introduce a procedure that constructs a large fractional matching over the realized edges of $S$ and use this to argue that there must exist a heavy realized matching among the edges in $S$.

For simplicity of the analysis, we assume that for any realization $\mathcal{G}=(V, \mathcal{E})$ of $G$, the maximum weighted matching denoted by $\matching{\mathcal{E}}$ is unique. This can be guaranteed by either using a deterministic algorithm for finding the matching $\matching{\mathcal{E}}$ or initially perturbing the edge weights by sufficiently small factors so that the maximum weighted matching becomes unique. Having this, we start with the following definition.

\begin{definition}
For any edge $e$, we denote by $q_e := \Pr_{E_p}[e \in \matching{E_p}]$ the probability with which $e$ appears in the (unique) maximum weighted matching of realization $E_p$. We refer to $q_e$ as the {\em matching probability} of edge $e$. Moreover, for any edge subset $F \subseteq E$, we denote by $q(F) := \sum_{e \in F} q_e$ the sum of matching probabilities of the edges in $F$.

We further use $\qw_e$ to denote $q_e \cdot w_e$ and use $\qw(F)$ to denote $\sum_{e \in F}\qw_e$. We call $\qw_e$ (resp. $\qw(F)$) the {\em expected matching weight} of $e$ (resp. $F$).
\end{definition}

Now, based on their matching probabilities, we partition the edges into two sets of {\em crucial} and {\em non-crucial} edges.

\begin{definition}[Crucial and non-crucial edges]
	For threshold $\tau = \frac{\epsilon^3 p}{20 \log(1/\epsilon)}$, we call any edge with $q_e < \tau$ a {\em non-crucial} edge and any edge with $q_e \geq \tau$ a {\em crucial} edge. We denote by $N$ the set of all non-crucial edges in $E$ and denote by $C$ the set of all crucial edges in $E$.
\end{definition}

We start with a couple of simple observations that will help both in gaining more insights on the definitions above and will be useful in our proofs later.

\begin{observation}\label{obs:optisqwhqwl}
	$\opt = \qw(N) + \qw(C).$
\end{observation}
\begin{proof}
	By definition, we know $\opt = \sum_{e \in E} q_e \cdot w_e = \sum_{e \in E} \qw_e$. Since $E = C \cup N$ and $C \cap N = \emptyset$, we have $\opt =\sum_{e \in N} \qw_e + \sum_{e \in C} \qw_e = \qw(N) + \qw(C)$.
\end{proof}

\begin{observation}\label{obs:samplingprob}
	An edge $e \in E$ is chosen to be in set $S$ by Algorithm~\ref{alg:nonadaptive} with probability exactly $1-(1-q_e)^R$.
\end{observation}
\begin{proof}
	In each iteration of Algorithm~\ref{alg:nonadaptive} edge $e$ appears in the maximum weighted matching $\matching{\mathcal{G}_i}$ with probability exactly $q_e$. Since Algorithm~\ref{alg:nonadaptive} is composed of $R$ independent iterations (i.e., the realizations $\mathcal{G}_i$ picked at different rounds are independent of each other), the probability that edge $e$ is not picked in any of these rounds is $(1-q_e)^R$ and therefore it appears in $S$ with probability $1-(1-q_e)^R$.
\end{proof}

As demonstrated by Observation~\ref{obs:samplingprob}, the crucial edges have a higher chance of appearing in the sample $S$. In fact, each crucial edge is sampled in each iteration of Algorithm~\ref{alg:nonadaptive} with probability at least $\tau$ and the number of iterations $R$ of Algorithm~\ref{alg:nonadaptive} is much larger than $1/\tau$; thus we expect almost every crucial edge to be sampled in $S$. We formalize this intuition in the following lemma whose proof we defer to Appendix~\ref{sec:otheromitted}.

\begin{lemma}[crucial edges lemma]\label{lem:sampleallcrucial}
	Let $S$ be the sample obtained by Algorithm~\ref{alg:nonadaptive}. Then, we have $\E[\qw(S \cap C)] \geq (1-\epsilon) \qw(C).$
\end{lemma}

\begin{observation}\label{obs:expmatchinggtq}
	$\expmatching{S} \geq \qw(S)$.
\end{observation}
\begin{proof}
	Consider $\mu = \matching{E_p} \cap S$, which is clearly a valid realized matching of $S$. It suffices to show that $\E[\text{weight of }\mu] \geq \qw(S)$. Note that any edge $e \in S$ that appears in $\matching{E_p}$ will appear in $\mu$, therefore, each edge $e \in S$ appears in $\mu$ with probability $q_e$. This means that $\E[\text{weight of }\mu] = \sum_{e\in S} q_e \cdot w_e = \qw(S)$ as desired.	
\end{proof}

The combination of Lemma~\ref{lem:sampleallcrucial} and Observation~\ref{obs:expmatchinggtq} implies that Algorithm~\ref{alg:nonadaptive} achieves an expected matching of weight at least $(1-\epsilon)\qw(C)$. This implies that if $\qw(C)$ is sufficiently close to $\opt$ (which is equvialent to $\qw(C) + \qw(N)$ by Observation~\ref{obs:optisqwhqwl}), Algorithm~\ref{alg:nonadaptive} obtains a good approximation. However, it might be the case that indeed the expected weight $\qw(C)$ of the crucial edges is very small or even $0$ with $\qw(N)$ being close to $\opt$. To handle this, we need a different argument for non-crucial edges. The challenge is that the matching probability of a non-crucial edge can be arbitrarily small, and may even depend on $n$. Consider for example the complete bipartite graph $G_{n, n}$ with all edge weights of 1 (i.e., the graph is unweighted). One can show that the expected matching of $G_{n, n}$ is as large as $n - o(1)$ with high probability (see e.g., \cite{DBLP:conf/sigecom/BlumDHPSS15}) while the matching probability of every edge\footnote{Here for the sake of this example, we assume that the algorithm to obtain the maximum matching of a realization of $G_{2n}$ is not biased towards including any specific edge.} in $G_{n, n}$ is roughly $\sfrac{1}{n}$. Therefore, since $S$ is of  constant degree, $q(S)$ will not be even a constant fraction of $n - o(1)$ and we cannot use Observation~\ref{obs:expmatchinggtq} to argue that the \expmatching{S} is large.

To alleviate the above-mentioned problem, we need to be able to get a large matching among the non-crucial edges too. This is the issue that we address next.

\paragraph{A lemma for non-crucial edges.} We describe a procedure -- formalized as Procedure~\ref{proc:non-crucial} -- to construct a heavy fractional matching on the realized portion of the non-crucial edges $S \cap E_p$ of $S$ which also enjoys some other properties of interest. For simplicity of notation, we use $S_p$ to denote $S \cap E_p$.

\begin{procedure}{Constructs a fractional matching  $x^N$ on non-crucial realized edges of $S$.
}\label{proc:non-crucial}
For any edge $e \in S_p$ initially set $\tilde x^N_e \gets 0$. Then update $\tilde x^N$ as follows:
\begin{enumerate}[label={(\arabic*)}]
	\item For any realized sampled non-crucial edge $e$ (i.e., $e \in S_p \cap N$), set $\tilde x^N_e \gets \min\{f_e/p, 2\tau/p \}$ where $f_e$ denotes the fraction of iterations of Algorithm~\ref{alg:nonadaptive} in which edge $e$ is part of the picked matching $M(\mathcal{E}_r)$.
	\item Initially set the {\em scaling-factor} $s_e$ of each edge $e$ to be $s_e = 1$. Then loop over the vertices $v \in V$ in an arbitrary order and for any $e$ incident to $v$, update $$s_e \gets \min\Big\{s_e, \max\{q^N_v, \epsilon\}/\tilde x^N_v\Big\},$$ where $q^N_v := \sum_{e: e\in N, v \in e}q_e$ denotes the {\em non-crucial-weight} of vertex $v$.
	\item Scale down the fractional matching in the following way: for any edge $e$, set $x^N_e \gets \tilde x^N_e \cdot s_e$.
\end{enumerate}
\end{procedure}

The following lemma highlights the properties of the procedure above.

\begin{lemma}[non-crucial edges lemma]\label{lem:non-cruciallemma}
The fractional matching $x^N$ obtained by Procedure~\ref{proc:non-crucial} has the following properties:
	{
		\setlength{\abovedisplayskip}{5pt}
		\setlength{\belowdisplayskip}{2pt}
		\begin{enumerate}
			\item For any $U \subseteq V$ with $|U| \le 1/\epsilon$, $x^N$ fills only $\epsilon$ fraction of its blossom inequality. That is, 		
				\begin{equation*}
					x^N(U) \leq \epsilon \lfloor |U|/2 \rfloor \qquad \forall U \subseteq V: |U| \leq 1/\epsilon.
				\end{equation*}
			\item The {\em non-crucial budgets} of the vertices are (almost) preserved. More precisely, $$x^N_v \leq \max \{q^N_v, \epsilon\} \qquad \forall v \in V.$$
			\item The expected weight of the fractional matching is sufficiently close to that of non-crucial edges, i.e., $$\E\bigg[\sum_{e \in S_p \cap N} x^N_e \cdot w_e \bigg]\geq (1-10\epsilon) \qw(N).$$
		\end{enumerate}
	}
\end{lemma}

\paragraph{The intuition behind Procedure~\ref{proc:non-crucial}.} Observe that the fractional matching constructed by Procedure~\ref{proc:non-crucial} relies critically on $f_e$, the fraction of iterations in which edge $e$ is sampled by the algorithm. Recall that the probability with which Algorithm~\ref{alg:nonadaptive} samples an edge $e$ is precisely equal to $q_e$. Therefore it is not hard to see that $\E[f_e] = q_e$. Similar to $q$, we can see the collection of $f_e$'s on all edges as a fractional matching. In this regard, since $\E[f_e] = q_e$, we have $\E[\sum_{e \in N} f_e \cdot w_e] = \E[q_{e \in N} \cdot w_e] = \qw(N)$. Despite these similarities, note that by definition, $f$ is non-zero only on the edges sampled by Algorithm~\ref{alg:nonadaptive}. This is desirable since we want to construct a large fractional matching only on the sampled edges. However, we further want our fractional matching to be non-zero only on the {\em realized} sampled edges. To do this, the final fractional matching $x$ that we construct is roughly as follows: $x_e$ is $f_e / p$ if $e$ is realized and it is 0 otherwise. Since each edge is realized with probability $p$, we have $\E[x_e] = p \cdot (f_e/p) + (1-p) \cdot 0 = f_e$. Note, however, that we have to make sure that $x$ is a valid fractional matching. That is, $x$ should not assign a fractional matching of larger than 1 to any vertex. (Properties~1 and 2 even impose  stricter restrictions) To do this, we may have to {\em manually} scale down the value of $x$ after observing the realization. However, we need to argue that this does not hurt the total size of it by a significant factor. For this, we use the fact that $f_e$ for most non-crucial edges is very small due to its value being close to $q_e$ which is at most $\tau$ for all non-crucial edges. This, combined with the independence of edge realizations, indicates e.g., that it is very unlikely that $x$ exceeds 1 by a larger than $1+\epsilon$ factor. Note that unfortunately the same procedure does not provide a good approximation on the crucial edges. The reason is that for crucial edges, $f_e$ can be as large as $p$ and the probability that $x$ exceeds 1 will not negligible.

As for the proof of Lemma~\ref{lem:non-cruciallemma}, note that the first and the second properties are directly satisfied by Procedure~\ref{proc:non-crucial}. To see this, observe that for any edge $e$, we have $x^N_e \leq 2 \tau/p \le \epsilon^3$. This means that for any subset $U$ of the vertices, we have
\begin{align*}
x^N(U) &\leq \epsilon^3 \cdot \binom{|U|}{2} = \epsilon^2 \cdot \dfrac{|U| \cdot \big(|U|-1\big)}{2},
\end{align*}
which implies for any $U$ with $|U| \leq 1/\epsilon$, that
\begin{align*}
x^N(U) &\leq \epsilon^3 \cdot \dfrac{|U|\cdot\big(|U|-1\big)}{2} \leq \epsilon^2 \cdot \dfrac{|U|-1}{2} \leq \epsilon^2 \big\lfloor |U|/2 \big\rfloor \leq \epsilon \big\lfloor |U|/2 \big\rfloor,
\end{align*}
completing the proof of property 1. Property 2 is also simple to prove. In fact, steps 2 and 3 of Procedure~\ref{proc:non-crucial} are solely written to satisfy this property. To see this, take a vertex $v$, if $\tilde x^N_v \leq \max\{q^N_v, \epsilon\}$ the non-crucial budget of $v$ is preserved since the scaling-factors are no more than 1. Otherwise, by the end of step 2 we ensure that for any edge incident to $v$ we have $s_e \leq \max\{q^N_v, \epsilon\}/\tilde x^N_v$. Thus, once completing step 3, we have
\begin{equation*}
	x^N_v = \tilde x^N_v \cdot s_e \leq \tilde x^N_v \cdot \max\{q^N_v, \epsilon\}/\tilde x^N_v = \max\{q^N_v, \epsilon\},
\end{equation*}
which is the desired bound for property 2. It only remains to prove that the fractional matching assigned to the realized sampled non-crucial edges is large as required by property 3. The proof of this part is rather technical and to prevent interruptions to the flow of the paper, we defer it to Appendix~\ref{sec:missingproofs}.

%
%

\paragraph{Implications.} By coupling Lemma~\ref{lem:sampleallcrucial} and Lemma~\ref{lem:non-cruciallemma} we immediately get an analysis that ensures Algorithm~\ref{alg:nonadaptive} obtains an (almost) $1/2$ approximation. To see this, recall by Observation~\ref{obs:optisqwhqwl} that $\opt = \qw(C) + \qw(N)$, thus, either $\qw(C) \geq \opt/2$ or $\qw(N) \geq \opt/2$. If $\qw(C) \geq \opt/2$, then Lemma~\ref{lem:sampleallcrucial} implies that the expected matching weight of our sample is at least $(1-\epsilon)\opt/2$. On the other hand, if $\qw(N) \geq \opt/2$, the fractional matching obtained by Lemma~\ref{lem:non-cruciallemma} which also satisfies blossom inequalities, implies that an integral matching of size at least $(1-10\epsilon)\opt/2$ must exist in the realization.

\begin{corollary}\label{cor:halfapprox}
	For any desirably small $\epsilon$, Algorithm~\ref{alg:nonadaptive} provides a $(\sfrac{1}{2}-\epsilon)$ approximation for weighted graphs by querying $\Ot{1/\epsilon^4 p}$ edges per vertex.
\end{corollary}

Note that Corollary~\ref{cor:halfapprox} already improves the number of per-vertex queries of known results for weighted graphs due to \cite{DBLP:conf/sigecom/BehnezhadR18, DBLP:conf/soda/YamaguchiM18}. Our goal, however, is to provide a much better guarantee on the approximation factor. Suppose for example, that $\qw(N) = \qw(C) = \opt/2$. In this case, to achieve any approximation factor better than $1/2$, we need to argue that the crucial edges and the non-crucial edges can augment each other to obtain a matching that is much heavier than what they achieve individually. This is the issue that we address in the next two sections.

\input{augmentation}

\input{augmentation-weighted}

%% file: augmentation.tex
\section{Beyond Half Approximation -- Unweighted Graphs}
In this section, we devise a process that constructs a large fractional matching on the realized graph by assigning values to both crucial and non-crucial edges. For non-crucial edges, we follow Procedure~\ref{proc:non-crucial} in obtaining the fractional matching. For crucial edges, however, we take a different approach in constructing the fractional matching. Before describing the actual procedure, we emphasize on the following property of Procedure~\ref{proc:non-crucial} which is necessary for augmenting it with crucial edges.

\begin{observation}
	Procedure~\ref{proc:non-crucial} does not look at how the crucial edges are realized.
\end{observation}

Intuitively, the observation above tells us that the large fractional matching that we obtain on realized non-crucial edges does not adversarially affect the realization of crucial edges since Procedure~\ref{proc:non-crucial} is essentially unaware of the realization of crucial edges. As such, if we are able to construct a large realized fractional matching on the crucial edges, that also (1) does not violate the crucial budget of the vertices, or the blossom inequalities, and that (2) does not ``look" at the realization of the non-crucial edges, we can plug the two fractional matchings together to obtain a valid fractional matching that combines both non-crucial and crucial edges. This is, unfortunately, not possible on the crucial edges and the main obstacle is preserving the per-vertex budgets.

To illustrate the above-mentioned problem, consider a graph with $2n$ vertices and $n$ edges where each vertex is connected to exactly one edge, i.e., the graph is a matching of size $n$. Any of these edges that is realized will be part of the realized matching, thus, for any edge $e$ in this graph we have $q_e = p$; which means they are all crucial edges and we have $\qw(C) = pn$. Note that the crucial budget $q^C_v$ of each of the vertices is $p$. Therefore, if we want to preserve these crucial budgets on the realized crucial edges, the fractional value that we assign to each realized edge would be at most $p$ (instead of 1); implying that the expected fractional matching that we get would have a total weight of $p^2n$ in expectation which is only a $p$ fraction of $\qw(C)$.

Recall that preserving the crucial/non-crucial per-vertex budgets was to ensure that once we combine the crucial and non-crucial fractional matchings, the total fractional matching connected to each vertex does not exceed 1. To achieve this, a slightly weaker constraint is also sufficient. Consider a vertex $v$ with non-crucial budget $q^N_v$ and crucial budget $q^C_v$. If $q^N_v + q^C_v$ (i.e., $q_v$) is much smaller than 1, we can allow the crucial fractional matching to assign a value of (roughly) up to $1-q^N_v$ to the edges connected to $v$. This, for instance, resolves the issue of the example in the previous paragraph. Thus, it only remains to argue that one can find a large such fractional matching on realized crucial edges. We formalize the procedure for doing this as Procedure~\ref{proc:crucial}.

\begin{procedure}{Constructing a fractional matching $x^C$ for unweighted graphs on the realized crucial edges of $S$.}\label{proc:crucial}
\begin{procinput}The realized portion $R^C := S_p \cap C$ of the sampled crucial edges.	
\end{procinput}

For any matching $\mu \in S_p \cap C$ define the {\em appearance-probability} $q(\mu | R^C)$ of $\mu$ to be the probability with which $\mu$ is the portion of $S_p \cap C$ that appears in the omniscient optimum, given the realization $R^C$ of the crucial edges. Formally,
\begin{equation*}
	q(\mu|R^C) = \Pr\Big[\mu = \big(\matching{E_p} \cap S_p \cap C\big) \Big| E_p \cap C =  R^C\Big].	
\end{equation*}
Among all matchings in $S_p \cap C$, we draw one according to the appearance-probabilities. Let us denote this matching by $\mu^C$. For any edge $e=(u, v) \in \mu^C$, set $$x^C_e \gets (1-\epsilon) \min \big\{ 1 - q^N_v, 1-q^N_u  \big\},$$
and for any other edge $e \in S_p \cap C$ we set $x^C_e \gets 0$.
\end{procedure}

We first show that by combining Procedures~\ref{proc:non-crucial} and \ref{proc:crucial} we can obtain a $\approx 0.6568$ approximation for unweighted graphs. Define fractional matching $x$ as follows
\begin{equation}
	x_e := x^N_e \qquad \forall e \in N, \qquad\qquad x_e := x^C_e \qquad \forall e \in C.
\end{equation}

\begin{claim}
	$x$ is a valid fractional matching that satisfies blossom inequalities of size up to $1/\epsilon$.
\end{claim}
\begin{proof}
	Fix any arbitrary subset $U \subseteq V$ of size at most $1/\epsilon$. Lemma~\ref{lem:non-cruciallemma} guarantees that the fractional matching on non-crucial edges of $U$ has size at most $\epsilon \lfloor\frac{|U|-1}{2} \rfloor$. On the other hand, since $\mu^C$ is an integral matching, it has at most $\lfloor \frac{|U|-1}{2} \rfloor$ edges in $U$. Since the fractional matching that we assign each edge of $\mu^C$ is at most $1-\epsilon$, overall the total size of the fractional matching assigned to the edges in $U$ cannot be more than $\epsilon \lfloor\frac{|U|-1}{2} \rfloor + (1-\epsilon) \lfloor \frac{|U|-1}{2} \rfloor = \lfloor\frac{|U|-1}{2} \rfloor$.
\end{proof}

\begin{theorem}\label{thm:nonadaptiveunweighted}
	If $G$ is unweighted, the constructed fractional matching $x$ of Procedure~\ref{proc:crucial} has size $\E\big[\sum_e x_e\big] \geq (1-2\epsilon)(4\sqrt{2}-5)\opt$. Therefore, Algorithm~\ref{alg:nonadaptive}, in expectation, achieves an approximation factor of at least $(1-2\epsilon)(4\sqrt{2}-5)$.
\end{theorem}
\begin{proof}
	Let us denote by $\alg := \sum_e x_e$ the size of our fractional matching $x$. We know by definition that $\alg = \sum_{e \in N} x^N_e + \sum_{e \in C} x^C_e$. It can be deducted by property 3 of Lemma~\ref{lem:non-cruciallemma} that 
	\begin{equation}\label{eq:p2non-crucial}
		\E\Big[\sum_{e \in N}x^N_e\Big] \geq (1-\epsilon) \qw(N) = (1-\epsilon) q(N),
	\end{equation}
	where the latter equality is due to the assumption that the graph is unweighted. Our goal, now, is to show that $\E[\sum_{e \in C}x^C_e]$ is also large. Take a crucial edge $e=(u, v)$, we know that Algorithm~\ref{alg:nonadaptive} picks $e$ with probability at least $1-\epsilon$ since $e$ is a crucial edge. Assuming that $e$ is picked by Algorithm~\ref{alg:nonadaptive}, $e$ is part of the matching $\mu^C$ picked by Procedure~\ref{proc:crucial} with probability at least $q_e$. And if $e$ is part of $\mu^C$, the fractional matching that will be assigned to it is $(1-\epsilon)\min\{1-q^N_v, 1-q^N_u\}$. Thus, for any crucial edge $e=(u, v)$, we have
	\begin{equation*}
		\E\big[x^C_e\big] = (1-\epsilon)\cdot q_e \cdot (1-\epsilon)\min\{1-q^N_v, 1-q^N_u\} \geq (1-2\epsilon)q_e \cdot \min\{1-q^N_v, 1-q^N_u\}.
	\end{equation*}
	To get rid of the minimization above, we make the crucial edges directed towards their endpoint with the higher non-crucial budget. Formally, a crucial edge $e = (u, v)$ is directed towards its endpoint $u$ if $q^N_u > q^N_v$ and in case of a tie (i.e., if $q^N_u = q^N_v$), we break it arbitrarily. For any vertex $v$ we denote its incoming crucial edges by $N^{C-}(v)$ and use $q^{C-}_v := \sum_{u \in N^{C-}(v)} q_{(u, v)}$ to denote the total matching probabilities of the edges that are directed towards $v$. With these definitions, we have
	\begin{align}\label{eq:p2h}
		\nonumber\E\Big[\sum_{e \in C} x^C_e \Big] &= \sum_v (1-2\epsilon) (1-q^N_v)q^{C-}_v\\ 
		\nonumber &= (1-2\epsilon)\sum_v\big(q^{C-}_v - q^N_v q^{C-}_v\big) \\
		\nonumber &= (1-2\epsilon)\sum_v q^{C-}_v - (1-2\epsilon) \sum_v q^N_v q^{C-}_v \\
		&= (1-2\epsilon)q(C) - (1-2\epsilon)\sum_v q^N_v q^{C-}_v.
	\end{align}
	Combining (\ref{eq:p2non-crucial}) and (\ref{eq:p2h}) we get
	\begin{align*}
		\E[\alg] = \E\big[\sum_{e \in N} x^N_e \big] + \E[\sum_{e \in C} x^C_e \big] &\geq (1-\epsilon)q(N) + (1-2\epsilon)q(C) - (1-2\epsilon)\sum_v q^N_v q^{C-}_v,\\
		&\geq (1-2\epsilon) \bigg( q(N) + q(C) - \sum_v q^N_v q^{C-}_v\bigg).
	\end{align*}
	On the other hand, recall that $\opt = q(N) + q(C)$, thus we have
	\begin{equation}\label{eq:XBST}
		\frac{\E[\alg]}{\opt} \geq \frac{(1-2\epsilon)\Big( q(N) + q(C) - \sum_v q^N_v q^{C-}_v\Big)}{q(N)+q(C)} \geq (1-2\epsilon)\bigg( 1 - \frac{\sum_v q^N_v q^{C-}_v}{q(N)+q(C)}\bigg).
	\end{equation}
	Note that since each crucial edge is directed towards exactly one of its endpoints, we have $q(C) = \sum_v q^{C-}_v$. On the other hand, we have $\sum_v q^N_v = 2q(N)$ since the matching probability of each non-crucial edge $(u, v)$ will contribute both to $q^N_v$ and $q^N_u$. Combining these two observations, we have 
	\begin{equation}\label{eq:CTAU}
		q(N) + q(C) = \sum_v q^{C-}_v + \frac{q^N_v}{2}.
	\end{equation}
	Combining (\ref{eq:XBST}) and (\ref{eq:CTAU}) we get
	\begin{equation}\label{eq:XBGC}
		\frac{\E[\alg]}{\opt} \geq (1-2\epsilon)\bigg( 1 - \frac{\sum_v q^N_v q^{C-}_v}{\sum_v q^{C-}_v + \frac{q^N_v}{2}}\bigg).
	\end{equation}
	We use the following mathematical lemma to show the desired bound on this ratio.
	\begin{lemma}\label{lem:mathratio}
		Given any set of numbers $a_1, \ldots, a_n$ and $b_1, \ldots, b_n$ such that
		\begin{enumerate}[label=(\roman*)]
			\item $a_i \geq 0$, $b_i \geq 0$, and $a_i + b_i \leq 1$ for any $i\in[n]$, and
			\item $\sum_{i=1}^N a_i + b_i > 0$,
		\end{enumerate}
		we have $\frac{\sum_{i=1}^{n}a_ib_i}{\sum_{i=1}^{n} a_i + \frac{b_i}{2}} \leq 6-4\sqrt{2}.$
	\end{lemma}
	
	For any vertex $v$ we have $q^N_v \in (0, 1)$ and $q^{C-}_v \in (0, 1)$ and clearly $q^N_v + q^{C-}_v \leq 1$ since $q$ is a valid fractional matching and the amount of matching incident to each vertex is at most 1, therefore, condition (i) of Lemma~\ref{lem:mathratio} is satisfied. Furthermore, condition (ii) of Lemma~\ref{lem:mathratio} also holds so long as $\opt > 0$ which is always the case unless the graph is empty, thus we have
	\begin{equation*}
		\frac{\sum_v q^N_v q^{C-}_v}{\sum_v q^{C-}_v + \frac{q^N_v}{2}} \leq 6-4\sqrt{2}, \qquad \text{therefore,} \qquad  1- \frac{\sum_v q^N_v q^{C-}_v}{\sum_v q^{C-}_v + \frac{q^N_v}{2}} \geq 1-(6-4\sqrt{2}) = 4\sqrt{2}-5.
	\end{equation*}
	Replacing this in Inequality (\ref{eq:XBGC}) we get $\frac{\E[\alg]}{\opt} \geq (1-2\epsilon)(4\sqrt{2}-5)$ or equivalently the desired bound in Theorem~\ref{thm:nonadaptiveunweighted} that $\E[\alg] \geq (1-2\epsilon)(4\sqrt{2}-5)\opt$.
\end{proof}

We next show that our analysis in Theorem~\ref{thm:nonadaptiveunweighted} for the fractional matching $x$ constructed via the above-mentioned procedures is tight.

\begin{lemma}\label{lem:unweightedtight}
	There exists a bipartite unweighted graph $G$, for which the fractional matching $x$ construct via Procedures~\ref{proc:non-crucial} and \ref{proc:crucial} has an approximation factor of less than $4\sqrt{2} - 5 + o(1)$.
\end{lemma}
\begin{proof}
	\begin{figure}
	  \centering
	  \includegraphics[scale=0.90]{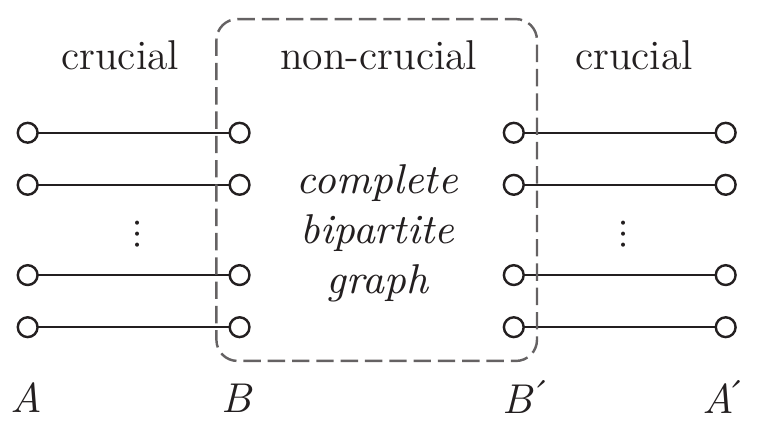}
	  \caption{An unweighted bipartite graph for which the fractional matching composed of Procedures~\ref{proc:non-crucial} and \ref{proc:crucial} does not provide a better than $4\sqrt{2}-5$ approximation.}
	  \label{fig:unweightedhardness}
	\end{figure}
	For a sufficiently large $L$, construct a graph $G'$ (refer to Figure~\ref{fig:unweightedhardness} for the illustration of the graph) with four sets $A, B, A', B'$ of $L$ vertices, i.e., the graph has $4L$ vertices in total. There is a complete bipartite graph between the vertices in $B$ and $B'$. There is also a perfect matching between $A$ and $B$ and a perfect matching between $B'$ and $A'$. Moreover, we set the realization probability $p$ of the graph to be $p=\sqrt{2}-1$. The optimal way of constructing a matching in a realization $G'_p$ of $G'$ is to first add all the realized edges between $A$ and $B$ or $A'$ and $B'$ to the matching; and then complement it via the realized edges between the unmatched vertices in $B$ and $B'$. Since there is a complete bipartite graph between the unmatched vertices in $B$ and $B'$, one can find a realized matching that is almost perfect. That is, this realized matching matches $1-o(1)$ fraction of the unmatched vertices in $B$ and $B'$. Thus, overall, we have
	\begin{align*}
		\E[\opt] &= \underbrace{p \times 2L}_{\text{matching between $A$ and $B$ or between $A'$ and $B'$}} + \underbrace{(1-o(1))(1-p)L}_{\text{matching between $B$ and $B'$}}\\
		&\geq (1+p-o(1))L\\
		&\geq (\sqrt{2}-o(1))L.
	\end{align*}
	The crucial edges of $G$ are those between $A$ and $B$ and those between $A'$ and $B'$. The rest of the edges are non-crucial. We have $\qw(C) = 2pL = (2\sqrt{2}-2)L$ and we have $\qw(N) = (2-\sqrt{2}-o(1))L$. Thus, the non-crucial budget of each vertex in $B$ or $B'$, which is $\qw(N)/L$, is equal to $(2-\sqrt{2}-o(1))$.

	The fractional matching that we construct by combining Procedures~\ref{proc:non-crucial} and \ref{proc:crucial} first obtains a fractional matching of size $(1-\epsilon)\qw(L)$ on the non-crucial edges. However, on each of the crucial edges $e=(u, v)$ that are realized, it puts a fractional matching of size $$(1-\epsilon)\min\{1-q^N_u, 1-q^N_v\} = (1-\epsilon)\Big(1-\big(2-\sqrt{2}-o(1)\big)\Big) = (1-\epsilon)(\sqrt{2}-1+o(1)).$$
	Meaning that overall, we construct a fractional matching of size only $(1-\epsilon)\Big(\sqrt{2}-1+o(1)\Big)\cdot p \cdot 2L = (1-\epsilon)(6-4\sqrt{2}+o(1))L$ on the crucial edges. Overall, the approximation factor would be
	\begin{align*}
		\frac{\E[\alg]}{\E[\opt]} &= \frac{(1-\epsilon)\Big(\overbrace{(2-\sqrt{2}-o(1))L}^{\text{Procedure~\ref{proc:non-crucial}}} + \overbrace{(6-4\sqrt{2}+o(1))L}^{\text{Procedure~\ref{proc:crucial}}} \Big)}{(\sqrt{2}-o(1))L}\\
		&\leq \frac{(1-\epsilon)\big(8-5\sqrt{2}+o(1)\big)}{\sqrt{2}-o(1)}\\
		&\leq (1-\epsilon)\big(4\sqrt{2}-5+o(1)\big).
	\end{align*}
	This completes the proof and almost matches the guarantee provided by Theorem~\ref{thm:nonadaptiveunweighted}.
\end{proof}

%% file: augmentation-weighted.tex
\section{Beyond Half Approximation -- Weighted Graphs}

We showed in the previous section that Procedure~\ref{proc:crucial} guarantees a $\approx 0.6568$-approximation for unweighted graphs. However, unfortunately, it does not provide anything better than a half approximation for weighted graphs. Recall by Corollary~\ref{cor:halfapprox} that we already achieve an almost half-approximation by combining Lemmas~\ref{lem:non-cruciallemma} and \ref{lem:sampleallcrucial}. Thus, Procedure~\ref{proc:crucial} does not have any benefits in the case of weighted graphs. In this section, we modify this procedure to bypass the half approximation barrier for weighted graphs.

\begin{figure}
  \centering
  \includegraphics[scale=0.95]{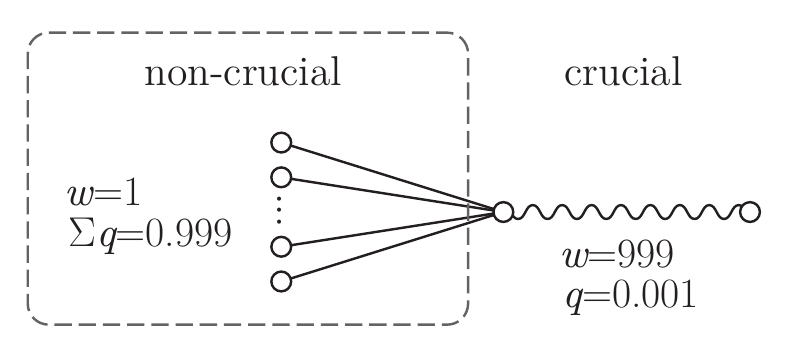}
  \caption{Example showing that Procedure~\ref{proc:crucial} does not provide a better than 0.5005 approximation on weighted graphs.}
  \label{fig:weightedapproxhard}
\end{figure}

We start the discussion of this section by an example that illustrates the main difficulty in the analysis of weighted graphs which also shows why Procedure~\ref{proc:crucial} does not provide a better than half approximation. Consider a star graph (Figure~\ref{fig:weightedapproxhard}) with one crucial edge $e$ of weight $w_e = 999$ and matching probability $q_e = 0.001$. The rest of the edges are non-crucial, each with a weight of 1 and sum of their matching probabilities is 0.999. These weights and probabilities are set in a way that makes the expected matching of both crucial and non-crucial edges equal (i.e., $\qw(C) = \qw(N) = 0.999$) while at the same time, assigning significantly different matching probabilities to them (observe that $q(C)=0.001$ while $q(N)=0.999$). The total expected matching of the graph is $\qw(C)+\qw(N) = 1.998$, however, the weight of the fractional matching obtained by Procedure~\ref{proc:crucial} is only\footnote{For clarity of exposition we hide the $1-\epsilon$ factors here.} 
\begin{equation*}
	\underbrace{0.999 \times 1}_{\text{from non-crucial edges}} + \underbrace{0.001}_{\text{\parbox{2.3cm}{probability that $e$ appears in $\mu^C$}}} \times  \underbrace{(1-0.999)}_{\text{\parbox{2.9cm}{budget remaining for $e$}}} \times 999 = 0.999999,
\end{equation*}
which provides only a $0.5005$-approximation. Using the same approach one can construct examples that show the approximation factor of Procedure~\ref{proc:crucial} is at most $0.5 + o(1)$.

\begin{remark}
	We remark that for weighted graphs, there is no procedure that allocates budgets to crucial and non-crucial edges prior to looking at the actual realizations, that has approximation factor better than $0.5 + o(1)$.
\end{remark}

To overcome the above-mentioned challenge, we devise a procedure that has dynamic budgets. That is, the procedure first looks at the realization of crucial edges, and then adjusts the budgets of non-crucial edges. Before delving into the details of the procedure, we describe how it is possible to obtain a near optimal approximation for the example of Figure~\ref{fig:weightedapproxhard}. Similar to the case of unweighted graphs, we can first use Procedure~\ref{proc:non-crucial} to construct a fractional matching on the non-crucial edges that does not violate the non-crucial budgets of the vertices. This provides a fractional matching of weight $\qw(N)$ and a half approximation. Next, we look at the realization of the crucial edges. If our crucial edge $e$ is not realized, then we report the fractional matching that we already have. However, if $e$ happens to be realized, we remove the fractional matching on the non-crucial edges and assign a fractional matching of 1 to edge $e$ which has a significantly higher weight. The expected weight of the fractional matching provided by this procedure is
\begin{equation*}
	\underbrace{0.999 \times (0.999 \cdot 1)}_{\text{if $e$ is not realized}} + \underbrace{0.001 \times 999}_{\text{if $e$ is realized}} \simeq 1.997,
\end{equation*}
which is very close to the expected matching of the original graph which is 1.998. The main intuition, here, was to {\em allow a crucial edge \underline{that is realized} to decrease the fractional matching on its incident non-crucial edges if that increases the total weight}. We formalize this approach in the following procedure and show that indeed it provides better than 0.5 approximation for weighted graphs. 

\begin{procedure}{Constructing a fractional matching $x$ for weighted graphs on the realized edges of $S$.}\label{proc:crucialweighted}
Consider the realization on sampled crucial edges and their realized portion $R^C := S_p \cap C$. Among all matchings in $S_p \cap C$, we draw one according to their appearance-probabilities based on $R^C$ (refer to Procedure~\ref{proc:crucial} for definition of appearance-probabilities). Let us denote this matching by $\mu^C$. For any edge $e=(u, v) \in \mu^C$, set $$x^C_e \gets (1-\epsilon) \argmax_{0 \le \alpha \le 1} \bigg( \frac{\min\{q^N_v, 1-\alpha \}}{q^N_v} \cdot \qw^N_v +  \frac{\min\{q^N_u, 1-\alpha \}}{q^N_u} \cdot \qw^N_u + \alpha \cdot w_e \bigg),$$
and for any other edge $e \in S_p \cap C$ we set $x^C_e \gets 0$.
\\\\
Let $x^N$ be the fractional matching of non-crucial edges constructed by the Procedure \ref{proc:non-crucial}. We define the fractional matching $x$ as follows.
\begin{align*}
	x_e := x^N_e \qquad \forall e \in N, \qquad\qquad x_e := x^C_e \qquad \forall e \in C.
\end{align*}
For any vertex $v$ with $x(v) > 1$, scale down the fractional matching on its non-crucial edges by an appropriate factor.
\end{procedure}

\begin{theorem}\label{thm:nonadaptiveweighted}
	Algorithm~\ref{alg:nonadaptive}, in expectation, provides a $\bu$ approximation for weighted graphs.
\end{theorem}
\begin{proof} 
	Recall that by Lemma \ref{lem:non-cruciallemma}, we know that Algorithm \ref{alg:nonadaptive} provides a fractional matching with an expected weight of at least $(1-10\epsilon) \qw(N)$. Also, it satisfies the blossom inequalities of size up to $1/\epsilon$. Therefore, by Lemma~\ref{lem:folklore}, the expected weight of the matching of this algorithm is $(1-11\epsilon) \qw(N)$. Also, by Lemma \ref{lem:sampleallcrucial}, the expected weight of the matching on only crucial edges is at least $(1-\epsilon) \qw(C)$. Since $\opt = \qw(C) + \qw(N)$, if at least one of $\qw(C)$ or $\qw(N)$ are at least $\bube \cdot \opt$, we can beat the $(\bube-11\epsilon)$ approximation factor and get $\bu$ approximation by choosing $\epsilon$ small enough. Otherwise, we have
	$$
		\bdbe \cdot \opt \le \qw(N), \qw(C) \le \bube \cdot \opt \,. 
	$$
	
	In this case, we show that the expected weight of the matching constructed by Algorithm~\ref{alg:nonadaptive} is at least $\bu \cdot \opt$. We first define two types of crucial edges and show that if the weight of these edges are greater than a specific threshold, Procedure \ref{proc:crucialweighted} produces a matching with the expected matching at least $\bu \cdot \opt$. Let $\delta=0.09$, we define these edges as follows.
	\begin{description}
		\item[Heavy edges.] We say that a crucial edge $e=(v,u) \in C$ is \textit{heavy} if $w_e \ge (1+\delta)(\qw^N_v + \qw^N_u)$. We use $H$ to denote the set of heavy edges. The weight of any heavy edge is larger than the sum of fractional matching of non-crucial edges of its both end. Therefore, in Procedure \ref{proc:crucialweighted}, a realized heavy edge reduces the fractional matching of non-crucial edges of its both ends to $0$, and we have $x_e^C = (1-\epsilon)$.
		\item[Semi-heavy edges.] We say that a non-heavy crucial edge $e=(v,u) \in (C \setminus H)$ is \textit{semi-heavy}, if for at least one of its endpoints, say w.l.o.g., vertex $v$, we have $w_e \ge 2(1+\delta) \qw^N_v$ and for its other endpoint we have $q^N_u \le (1-\delta)$ and $q^N_v \ge q^N_u$.  We use $H^\star$ to denote the set of semi-heavy edges. The weight of any semi-heavy edge is larger than the fractional matching of non-crucial edges of one of its endpoints. Therefore, it reduces the fractional matching of non-crucial edges on this endpoint. Formally, for any semi-heavy edge $e$ we have $x_e^C \ge (1-\epsilon) (1-q^N_u) \ge (1-\epsilon) \delta$.
	\end{description}
	In the following claim, we show that if a large ``portion" of critical edges are heavy or semi-heavy, we can construct a fractional matching with an expected weight of $ \bu \cdot \opt$.
	\begin{claim}
		If $\qw(H)+\qw(H^\star) \ge 0.09 \cdot \qw(C)$, then the expected weight of the matching produced by Algorithm \ref{alg:nonadaptive} is at least $ \bu \cdot \opt$.
	\end{claim}
	\begin{proof}
		Consider a heavy edge $e=(v,u) \in H$, if this edge realized, Procedure \ref{proc:crucialweighted} sets $x^C_e = (1-\epsilon)$, and removes the fractional matching of non-crucial edges of both ends. Therefore, it adds a weight of $(1-\epsilon) w_e - (\qw^N_v+\qw^N_u)$ to our fractional matching which is at least
		\begin{align}
			(1-\epsilon) w_e - (\qw^N_v+\qw^N_u) &\ge (1-\epsilon) w_e - \frac{1}{1+\delta} w_e 
			& \text{Since } e \text{ is heavy and } w_e \ge (1+\delta)(\qw^N_v+\qw^N_u). \nonumber \\
			& = (\frac{\delta}{1+\delta}-\epsilon) w_e \label{ineq:heavy} \,.
		\end{align}
		Moreover, suppose that $e'=(v',u') \in H^\star$ is a semi-heavy edge. By definition of semi-heavy edges, we know that for one of endpoints of $e'$, say $v'$, we have $w_e \ge 2(1+\delta) \qw^N_{v'}$, and for the other endpoint we have $q^N_{u'}\le (1-\delta)$ and $q^N_{v'} \ge q^N_{u'}$.  If edge $e'$ realized, it reduces the fractional matching of non-crucial edges of $v'$ to at most $q^N_{u'}$ and Procedure uses at least $(1-\epsilon) (1-q^N_{u'})$ fraction of the edge $e'$.  Therefore, the weight that it adds to the weight of the fractional matching produced by Procedure \ref{proc:crucialweighted} is at least  
			\begin{align}
			&(1-\epsilon) (1- q^N_{u'}) w_{e'} - (q^N_{v'}-q^N_{u'})\qw^N_{v'} \nonumber \\
			&\ge(1-\epsilon) (1-q^N_{u'}) w_{e'} - (1-q^N_{u'})\qw^N_{v'} \nonumber \\
			&=(1-q^N_{u'}) ((1-\epsilon)w_{e'} -\qw^N_{v'}) \nonumber \\
			&\ge (1-q^N_{u'}) ((1-\epsilon) w_{e'} - \frac{1}{2(1+\delta)} w_{e'}) 
			& \text{Since } e' \text{ is semi-heavy and } w_{e'} \ge 2(1+\delta)\qw^N_{v'} \nonumber \\
			& = (1-q^N_{u'}) \Big (\frac{1+2\delta}{2(1+\delta)}-\epsilon \Big ) w_{e'} \nonumber \\
			& \ge \delta \Big (\frac{1+2\delta}{2(1+\delta)}-\epsilon \Big ) w_{e'}
			& q^N_{u'} \le (1-\delta)  \nonumber\\
			&  \Big (\frac{\delta+2\delta^2}{2(1+\delta)}-\epsilon \Big ) w_{e'} \label{ineq:semi} \,.
		\end{align}
		It follows from inequalities (\ref{ineq:heavy}) and (\ref{ineq:semi}) that the weight of the expected matching is at least
		\begin{align*}
		&(1-10\epsilon) \qw(N) + (\frac{\delta}{1+\delta}-\epsilon) \qw(H) + (\frac{\delta+2\delta^2}{2(1+\delta)}-\epsilon) \qw(H^\star) \,.
		\end{align*}
		Since $\delta=0.09$, we have $\frac{\delta}{1+\delta} \ge 0.048$ and $\frac{\delta+2\delta^2}{2(1+\delta)} \ge 0.048$. Therefore, the expected weight of the fractional matching is at least
		\begin{align*}
		&(1-10\epsilon) \qw(N) + (0.048-\epsilon) (\qw(H)+\qw(H^\star))
		\\& \ge (1-10\epsilon) \qw(N) + (0.048-\epsilon) (0.09 \qw(C))
		\\& = (1-10\epsilon) (\opt - \qw(C)) + (0.00432 - \epsilon) \qw(C) & \qw(N)+\qw(C) = \opt.
		\\& \ge (1-10\epsilon) \opt - \qw(C) (1-0.0.00432)
		\\& \ge (1-10\epsilon) \opt - \bube \cdot \opt (1-0.00432)
		& \qw(C) \le \bube \cdot \opt
		\\&\ge (0.50106-10\epsilon) \cdot \opt \,,
		\end{align*}
		By Lemma \ref{lem:folklore}, we also lose a factor of $(1-\epsilon)$ to satisfy the blossom inequalities. Therefore, by choosing $\epsilon$ small enough, we can get  $\bu$ approximation which proves the claim.
	\end{proof}
By the previous claim, we know that if $\qw(H)+\qw(H^\star) \ge 0.09 \qw(C)$, we already get our desired $\bu$ approximation. Therefore, from now on, we assume that $\qw(H)+\qw(H^\star) < 0.09 \qw(C)$. Though, for ease of exposition, we do not explicitly mention this condition in the forthcoming statements. Define $C^\star := C \setminus (H \cup H^\star)$ to be the set of crucial edges that are not heavy or semi-heavy. We have $\qw(C^\star) \ge (1-0.09) \qw(C) = 0.91 \qw(C)$.
\begin{claim}
The expected weight of the matching returned by Algorithm \ref{alg:nonadaptive} is at least $ (1-2 \epsilon)0.551(\qw(L)+\qw(C^\star))-8 \epsilon\opt$.
\end{claim}
\begin{proof}
We partition the edges in $C^\star$ into three types, and according to these types, we make the edges directed towards one of their endpoints. Let $e=(v,u) \in C^\star$. W.l.o.g., assume that $q^N_v \ge q^N_u$. We define the following three types:
\begin{description}
\item[Type 1.] If $\qw^N_v \ge \qw^N_u$, this edge is type~$1$. In this case we direct $e$ towards $v$.
\item[Type 2.] If $\qw^N_v < \qw^N_u$ and $w_e \le 2(1+\delta) \qw^N_v$, this edge is type $2$. In this case we direct $e$ towards $v$.
\item[Type 3.] For any edge that is not of type 1 or 2, we have  $\qw^N_v < \qw^N_u$ and $w_e > 2(1+\delta) \qw^N_v$. These edges are type $3$, and we direct them towards $u$.
\end{description}
The following observation demonstrates a critical property of edge directions defined above.
\begin{observation}
\label{obs:heavybound}
Let $e=(v,u) \in C^\star$ be a crucial edge which is directed towards $v$. Then we have $w_e \le 2(1+\delta)\qw^N_v$.
\end{observation}
\begin{proof}
It suffices to show that this property holds for all three types of edges. For any type~1 edge $e=(u, v)$, we have $\qw^N_v \ge \qw^N_u$. Since $e$ is not heavy, we have
$$
w_e \le (1+\delta)(\qw^N_v + \qw^N_u) \le 2(1+\delta)\qw^N_v \,.
$$ 
For any type~2 edge $e=(u, v)$, we have our desired inequality $w_e \le 2(1+\delta) \qw^N_v$ automatically by definition. For any type~3 edge $e=(u, v)$, if $e$ is directed towards $v$, we have $\qw^N_v > \qw^N_u$. Also $e$ is not heavy, therefore
$$
w_e \le (1+\delta)(\qw^N_v + \qw^N_u) < 2(1+\delta)\qw^N_v,
$$
which completes the proof.
\end{proof} 

The following observation is also another important property of direction of edges .

\begin{observation}
\label{obs:qlimit}
Let $e=(v,u) \in C^\star$ be an edge such that $q^N_v \ge q^N_u$. If $e$ is directed towards $u$, we have
$$
q^N_v \le q^N_u + \delta.
$$
\end{observation}

\begin{proof}
Since $q^N_v \ge q^N_u$, the only case that we direct $e$ towards $u$ is when $e$ is a type $3$ edge. In this case $w_e > 2(1+\delta) \qw^N_v$. Since $e$ is not semi-heavy, we must have $q^N_u > (1-\delta)$. Therefore we get our desired bound that $q^N_v \le 1 < q^N_u+ \delta.$
\end{proof}

Let $\bar x$ be a fractional matching obtained by combining Procedures~\ref{proc:non-crucial} and \ref{proc:crucial}. More specifically, let $\bar x^N$ be the fractional matching of Procedure \ref{proc:non-crucial} on non-crucial edges and $\bar x^C$ be the fractional matching of Procedure \ref{proc:crucial} on crucial edges. That is,
\begin{align*}
	\bar x_e := \bar x^N_e \qquad \forall e \in N, \qquad\qquad \bar x_e := x^C_e \qquad \forall e \in C.
\end{align*}
We show that expected weight of fractional matching $\bar x$ is at least $0.543(\qw(L)+\qw(C^\star))$.

 For any vertex $v$ we denote its incoming crucial edges in $C^\star$ by $N^{C-}(v)$ and use $\qw^{C-}_v := \sum_{u \in N^{C-}(v)} \qw_{(u, v)}$ to denote the expected weight of the matching of the edges that are directed towards $v$. If a crucial edge $e=(v,u)$ is directed towards vertex $v$, the budget that this edge can have in Procedure \ref{proc:crucial} is $(1-\epsilon) (1- \max\{ q^N_v,q^N_u\})$. If $e$ is directed towards $v$, by Observation \ref{obs:qlimit}, this value is at least
$(1-\epsilon) (1-\delta-q^N_v)$.
Our algorithm picks each crucial edge with the probability at least $1-\epsilon$. Therefore, for crucial edges in $C^\star$ we have
\begin{align*}
\E \Big [ \sum_{e \in C^\star} \bar x^C_e \cdot w_e \Big] &\ge \sum_v (1-\epsilon) (1-\epsilon) (1-\delta- q^N_v) \qw^{C-}_v
\\&\ge (1-2\epsilon) \sum_v (1-\delta- q^N_v) \qw^{C-}_v
\\&= (1-2\epsilon) (1-\delta) \qw(C^\star) - (1-2\epsilon) \sum_v q^N_v \qw^{C-}_v  \,.
\end{align*}
Therefore, the weight of the matching returned by our algorithm is at least
\begin{align*}
\E [ \alg] &\ge \E [ \sum_{e \in N} \bar x^N_e \cdot w_e] + \E [ \sum_{e \in C^\star} \bar x^C_e \cdot w_e] \\
&\ge (1-10\epsilon) \qw(N)+ (1-2\epsilon) \Big((1-\delta)\qw(C^\star) - \sum_v q^N_v \qw^{C-}_v\Big) \,.
\end{align*}
Since $\qw(N) \le \opt$, we have
\begin{align*}
\E [ \alg] -8 \epsilon \opt \ge   (1-2\epsilon) \Big(\qw(N)+ (1-\delta)\qw(C^\star) - \sum_v q^N_v \qw^{C-}_v\Big) \,.
\end{align*}
Therefore,
\begin{align*}
\frac{\E[\alg]-8 \epsilon \opt}{\qw(N)+ \qw(C^\star)} &\geq \frac{(1-2\epsilon)\Big( \qw(N) + (1-\delta) \qw(C^\star) - \sum_v q^N_v \qw^{C-}_v\Big)}{\qw(N)+\qw(C^\star)}\\ &\geq (1-2\epsilon)\bigg( 1 - \frac{\sum_v \delta \cdot \qw^{C-}_v + q^N_v \qw^{C-}_v}{\qw(N)+\qw(C^\star)}\bigg) \\
&=(1-2\epsilon)\bigg( 1 - \frac{\sum_v \delta \cdot \qw^{C-}_v + q^N_v \qw^{C-}_v}{\sum_v \qw^{C-}_v + \frac{\qw^{N}_v}{2}}\bigg).
	\end{align*}

\begin{claim}
\label{clm:singlev}
For each vertex $v$, we have
\begin{align}
\label{eq:SV}
\dfrac{\delta \cdot \qw^{C-}_v + q^N_v \qw^{C-}_v}{\qw^{C-}_v + \frac{\qw^{N}_v}{2}} \le 0.449 \,.
\end{align}
\end{claim}
\begin{proof}
By Observation \ref{obs:heavybound}, we know that for each edge $e$ directed towards $u$, we have $w_e \le 2(1+\delta) \qw^N_v$. Therefore,
$$
\qw^{C-}_v \le 2(1+\delta) q^C_v \qw^N_v \le 2(1+\delta) (1-q^N_v) \qw^N_v \,.
$$ 
Since the left side of (\ref{eq:SV}) is increasing in $\qw^{C-}_v$, and we have $\qw^{C-}_v \le 2(1+\delta)(1-q^N_v) \qw^N_v$, it takes its maximum value when $\qw^{C-}_v = 2(1+\delta) (1-q^N_v) \qw^N_v$. Thus,
$$
\dfrac{\delta \cdot \qw^{C-}_v + q^N_v \qw^{C-}_v}{\qw^{C-}_v + \frac{\qw^{N}_v}{2}} 
\le \dfrac{2(1+\delta) (1-q^N_v) \qw^N_v (\delta + q^N_v)}{\qw^N_v(\frac{1}{2}+ 2(1+\delta) (1-q^N_v))} =  \dfrac{2(1+\delta) (1-q^N_v) (\delta + q^N_v)}{\frac{1}{2}+ 2(1+\delta) (1-q^N_v)}.
$$
We hide the tedious mathematical calculations here; however, by setting $\delta=0.1$, one can verify that the value above is at most 
$$
\frac{171-10 \sqrt{146}}{110} \le 0.457
$$
 for $0\le q^N_v \le 1$.
\end{proof}
We use the following simple observation to complete the proof of the claim.
\begin{observation}
For positive real values $a,b,c,d, \alpha$, suppose that $\frac{a}{b} \le \alpha$ and $\frac{c}{d} \le \alpha$. Then, $\frac{a+c}{b+d} \le \alpha$.
\end{observation}
Using the observation above and Claim \ref{clm:singlev}, we have
$$
\frac{\sum_v \delta \cdot \qw^{C-}_v + q^N_v \qw^{C-}_v}{\sum_v \qw^{C-}_v + \frac{\qw^{N}_v}{2}} \le 0.449 \,.
$$
Therefore, we have
\begin{align*}
\frac{\E[\alg]-8 \epsilon \opt}{\qw(N)+ \qw(C^\star)} & \ge (1-2\epsilon)\bigg( 1 - \frac{\sum_v \delta \cdot \qw^{C-}_v + q^N_v \qw^{C-}_v}{\sum_v \qw^{C-}_v + \frac{\qw^{N}_v}{2}}\bigg) \\
&\ge (1-2 \epsilon) (1-0.449) = (1-10 \epsilon) 0.551 \,.
\end{align*}
This implies that
	$$
	\E[\alg] \ge (1-10 \epsilon)0.551(\qw(N)+\qw(C^\star)) -8 \epsilon\opt\,,
	$$
	which is the desired bound.
\end{proof}
By the claim above, we have
	\begin{align*}
	\E[\alg] &\ge (1-2 \epsilon)0.551(\qw(N)+\qw(C^\star))-8 \epsilon\opt \\
	& \ge (1-2 \epsilon)(0.551 \cdot 0.91)(\qw(N)+\qw(C))-8 \epsilon\opt & \text{Since $\qw(C^\star) \ge 0.91 \qw(C)$.} \\
	  &\ge (1-2 \epsilon)(0.5014)(\qw(N)+\qw(C))-8 \epsilon\opt\\
	  &\ge(1-2 \epsilon)(0.5014)\opt-8 \epsilon\opt & \qw(N)+\qw(C)=\opt.\\
	  & \ge (0.5014-10\epsilon) \cdot \opt \,.
	\end{align*}
	Also, this fractional matching satisfies the blossom inequalities of size up to $1/\epsilon$. Therefore, by Lemma~\ref{lem:folklore}, the expected weight of the matching of this algorithm is $(1-\epsilon)(0.5014 -10 \epsilon) \opt$ and by setting $\epsilon$ small enough, it becomes at least $\bu \cdot \opt$.
\end{proof}

%% file: appendix-proofs.tex
\section{Appendix: Omitted Proofs}\label{sec:missingproofs}
\subsection{Proof of the Non-crucial Edges Lemma}

In this section, we provide the complete proof for Lemma~\ref{lem:non-cruciallemma}.

\begin{proof}[Proof of Lemma~\ref{lem:non-cruciallemma}]

Proof of the first and the second properties were already given in Section~\ref{sec:nonadaptive}. Here we prove the third property. We first start with the following claim.

\begin{claim}\label{claim:flargewhp}
	By the end of Algorithm~\ref{alg:nonadaptive}, we have $\E \Big[\sum_{e \in S \cap N} \min\{f_e, 2 \tau \} \cdot w_e \Big] \geq (1-\epsilon)\qw(N)$.
\end{claim}
\begin{proof}
	We can think of the values of $f_e$ in the following way: For any edge $e$, $f_e$ is initially 0; then after each round $r$ of Algorithm~\ref{alg:nonadaptive}, we pick a matching $M(\mathcal{E}_r)$ and for any edge in this matching we update $f_e$ to be $f_e + 1/R$. Clearly by the end of the algorithm, the value of $f_e$ will be equal to the fraction of the matchings picked by the algorithm that contains $e$ which is precisely the definition of $f_e$. To argue that $\sum_{e\in S \cap N} f_e \cdot w_e$ is large, it suffices to show that the average weight of the matchings that are picked by Algorithm \ref{alg:nonadaptive} is close to $\expmatching{E}$. Let $M_1, M_2, \cdots, M_R$ be the random variables denoting the weights of the non-crucial edges in the matchings picked in each round of Algorithm~\ref{alg:nonadaptive}. For each $M_i$, we have
$$\E[M_i] = \sum_{e \in S \cap N} q_e \cdot w_e =  \qw(N)\,,$$
Further let $\bar M := (M_1, M_2, \cdots, M_R)/R$. One can easily confirm via linearity of expectation that
$\E[\bar M] = \qw(N)$. 
Note also that by definition of $f_e$, we have $\bar M = \sum_{e\in S \cap N} f_e \cdot w_e$. Hence,
$$
\E \Big[\sum_{e \in S \cap N} f_e \cdot w_e \Big ] = \sum_{e \in S \cap N} \E [ f_e  ] \cdot w_e = \qw(N) \,.$$

Next, we show that for every non-crucial edge $e$, the probability of $f_e$ exceeding $2 \tau$ is very small. Its proof is derived from the independence of the realizations taken by Algorithm~\ref{alg:nonadaptive}, the definition of $f_e$, and the fact that for all non-crucial edges $q_e < \tau$. Also, we assume that $\epsilon$ is a  small number and we have $\epsilon \le e^{-1}$.

\begin{claim}\label{cl:probltepsq}
For any non-crucial edge $e$, $f_e$ exceeds $2 \tau$ with probability at most $\epsilon \cdot q_e$.
\end{claim}
\begin{proof}
For an edge $e$, let $X_1, X_2, \cdots, X_R$ be random variables such that $X_i$ is $1$ if $e$ is picked in the maximum matching of round $i$ of Algorithm~\ref{alg:nonadaptive}, and is $0$ otherwise. Then we have $E[X_i] = q_e$ for each $X_i$. Recall that $f_e$ is the fraction of the matchings picked by the algorithm that contains $e$. Therefore, $f_e$ is the average of $X_1, X_2, \cdots, X_R$, i.e., $f_e = \frac{1}{R}(X_1+X_2+\cdots+X_R)$. Also, we have $\E[f_e] = q_e$. Let $X= X_1+X_2+\cdots+X_R$. It follows that $X= f_e \cdot R$, and we have
\begin{align*}
P [ f_e \ge 2 \tau ] &= P [ f_e - \tau \ge \tau ] \\
&\le P[ f_e - q_e \ge \tau] & \text{Since } e \text{ is non-crucial and } q_e < \tau .
\\& = P\Big[ f_e - \E[f_e] \ge \tau \Big]  
\\& = P\Big[ R \cdot f_e - \E[R \cdot f_e] \ge R \cdot \tau \Big]
\\& = P\Big[ X - \E[X] \ge R \cdot \tau \Big] & \text{Since $X= f_e \cdot R$.}
\\ & \le \exp\Big(-\frac{R \cdot \tau \cdot \log\big(1+(R \cdot \tau)/\E[X]\big)}{2}\Big) & \text{By Chernoff bound\footnotemark.}
\\ & \le \exp\Big(-\frac{R \cdot \tau \cdot \log(1+\frac{\tau}{q_e})}{2}\Big) & \text{$\E[X]= q_e \cdot R$.}
\\ & \le \exp\Big(-50 \log(1/\epsilon p) \log(1+\frac{\tau}{q_e})\Big) &\text{Since } R \cdot \tau > 100 \log(1/\epsilon p).
\\ & = \dfrac{1}{\exp\Big(50 \log(1/\epsilon p) \log(1+\frac{\tau}{q_e})\Big)}
\\ & = \dfrac{1}{(1+\frac{\tau}{q_e})^{50 \log(1/\epsilon p)}}
\\ & \le \dfrac{1}{(1+\frac{\tau}{q_e})(1+\frac{\tau}{q_e})^{49 \log(1/\epsilon p)}}
& \text{Since $\epsilon \le e^{-1}$ and $\log(1/\epsilon) \ge 1$.}
\\ & \le \dfrac{1}{(1+\frac{\tau}{q_e})\cdot 2^{49 \log(1/\epsilon p)}}
& \text{Since $\tau > q_e$ and  $1+\frac{\tau}{q_e}>2$.}
\\ & = \dfrac{1}{(1+\frac{\tau}{q_e})\cdot \exp\big({49/\log(2) \cdot \log(1/\epsilon p)}\big)}
\\ & \le \dfrac{1}{(1+\frac{\tau}{q_e})\cdot \exp\big(30\log(1/\epsilon p)\big)}
\\ & \le \dfrac{1}{(1+\frac{\tau}{q_e})\cdot e^{10} \cdot \exp\big(20\log(1/\epsilon p)\big)}
& \text{Since $\epsilon \le e^{-1}$ and $\log(1/\epsilon) \ge 1$.}
\\ & \le \dfrac{1}{20 (1+\frac{\tau}{q_e}) \cdot \exp\big(5 \log(1/\epsilon p)\big)}
\\ & \le \dfrac{1}{20 (1+\frac{\tau}{q_e}) \cdot \log(1/\epsilon p) \cdot \exp\big(4 \log(1/\epsilon p)\big)} 
& \text{ Since $e^{x} \ge x$ for all real numbers $x$.}
\\ & = \dfrac{\epsilon^4 p^4}{20 (1+\frac{\tau}{q_e}) \cdot \log(1/\epsilon p)}
\\ & \le \dfrac{\epsilon \tau} {(1+\frac{\tau}{q_e})} & \text{Since } \tau= \frac{\epsilon^3 p}{20 \log(1/\epsilon)}.
\\ & = \dfrac{\epsilon \cdot \tau \cdot q_e} {\tau+q_e}
\\ & \le \dfrac{\epsilon \cdot \tau \cdot q_e} {\tau}
\\ & = \epsilon \cdot q_e
\end{align*}
\footnotetext{By Chernoff bound we have $P \Big[ X \ge (1+\delta) E[X]\Big] \le \exp(-\frac{\delta \log(1+\delta) E[X]}{2})$.}
which proves the claim.
\end{proof}

By the claim above, we know that with probability at least $1-\epsilon q_e$, we have $f_e \le 2 \tau$. It follows that
	\begin{equation}\label{ieq:fealgo}
	 \E \Big[ \min \{ f_e, 2 \tau \} \Big] \ge \Pr \Big[f_e \le 2 \tau \Big] \cdot \E \Big [f_e \, |\, f_e \le 2 \tau \Big ] \,.
	\end{equation}
	On the other hand, we have
	\begin{equation}
	\label{ieq:fesum}
	q_e = E[f_e] = P \Big[f_e \le 2 \tau \Big] E \Big [f_e | f_e \le 2 \tau \Big ] + P \Big[f_e > 2 \tau \Big] E \Big [f_e | f_e > 2 \tau \Big ]\,.
	\end{equation}
	Combining (\ref{ieq:fealgo}) and (\ref{ieq:fesum}) gives
	\begin{align*}
	\E \Big[ \min \{ f_e, 2 \tau \} \Big] - E[f_e] &\le  P \Big[f_e > 2 \tau \Big] E \Big [f_e | f_e > 2 \tau \Big ]  \\
	&\le (\epsilon \cdot q_e) E \Big [f_e | f_e > 2 \tau \Big ]\\
	&\le (\epsilon \cdot q_e) & f_e \text{ is at most } 1.\\
	&= \epsilon E[f_e] \,.
	\end{align*}
	Therefore $\E \Big[ \min \{ f_e, 2 \tau \} \Big] \ge (1-\epsilon) E[f_e]$, and we have
	\begin{align*}
	\E \Big[\sum_{e \in S \cap N} \min\{f_e, 2 \tau \} \cdot w_e \Big] &
	= \sum_{e \in S \cap N} \E \Big[\min\{f_e, 2 \tau \}\Big] \cdot w_e\\
	& \ge \sum_{e \in S \cap N} (1-\epsilon)\E[f_e] \cdot w_e\\
	& =  (1-\epsilon) \sum_{e \in S \cap N} \E[f_e] \cdot w_e \\
	&= (1-\epsilon)\qw(N),
	\end{align*}
	which is our desired bound.
\end{proof}

\begin{claim}
\label{clm:non-crucial}
	By the end of step 1 of Procedure~\ref{proc:non-crucial}, we have $\E \big[ \sum_{e \in S_p \cap N} \tilde x^N_e . w_e \big] \geq (1-\epsilon)\qw(N)$.
\end{claim}
\begin{proof}
	Note that for each edge $e \in S_p \cap N$, we assign $\min\{f_e/p, 2\tau /p \}$ to $\tilde x^N_e$ by the end of step 1. Thus,
	\begin{align*}
		\E \bigg[\sum_{e \in S_p \cap N} &w_e \cdot \min\{f_e/p, 2\tau /p\} \bigg]\\
		&= \frac{1}{p} \cdot \E\bigg[\sum_{e \in S_p \cap N} w_e \cdot \min\{f_e, 2\tau\}\bigg]\\
		&= \frac{1}{p} \cdot \E\bigg[\sum_{e \in S \cap N} w_e \cdot  \min\{f_e, 2\tau\} \cdot \mathbbm{1}_{E_p}(e)\bigg] & \text{($\mathbbm{1}_{E_p}(e) = 1$ if $e \in E_p$ and 0 otherwise.)} \\
		&= \frac{1}{p} \cdot \sum_{e \in S \cap N} \E\big[ w_e \cdot \min\{f_e, 2\tau\} \cdot \mathbbm{1}_{E_p}(e)\big] & \text{By linearity of expectation.}\\
		&= \frac{1}{p} \cdot \sum_{e \in S \cap N} w_e \cdot \E\big[\min\{f_e, 2\tau\}\big] \cdot \E\big[\mathbbm{1}_{E_p}(e)\big] & \text{\parbox{7cm}{\vspace{0.3cm}Since value of $f_e$ is independent of its realization.\vspace{0.3cm}}}\\
		&= \frac{1}{p} \cdot \sum_{e \in S \cap N} w_e \cdot  \E\big[\min\{f_e, 2\tau\}\big] \cdot p\\
		&= \sum_{e \in S \cap N} w_e \cdot \E\big[\min\{f_e, 2\tau\}\big].
	\end{align*}
	Recall by Claim~\ref{claim:flargewhp} that we have $\sum_{e \in S \cap N} w_e \cdot \E\big[\min\{f_e, 2\tau\}\big] \geq (1-\epsilon)\qw(N)$. Combining it with the inequality above, we get, 
	$$\E \big[ \sum_{e \in S_p \cap N} \tilde x^N_e . w_e \big] \geq (1-\epsilon)\qw(N)$$
	which is the desired bound.
\end{proof}
Considering the matchings picked by Algorithm \ref{alg:nonadaptive}, the expected weight of each of them is $\opt$. As we showed in the claim above by the end of step $1$ of Procedure \ref{proc:non-crucial}, we have
$$ 
\E \big[ \sum_{e \in S_p \cap N} \tilde x^N_e . w_e \big] \geq (1-\epsilon)\qw(N) \,.
$$ 
We claim that for every realized edge $e \in (S_p \cap N)$, the scaling-factor of this edge which is $s_e$ is at least $(1-5\epsilon)$ with probability at least $(1-4\epsilon)$. Formally, our claim is as follows.
\begin{claim}
\label{clm:non-crucialedgeschernoff}
Let $v$ be one of the end points of a realized edge $e \in (S_p \cap N)$, then with probability at least $1-2\epsilon$, we have 
$$\max\{q^N_v, \epsilon\}/\tilde x^N_v \ge 1-5\epsilon \,.$$
\end{claim}
\begin{proof}
Since edge $e$ is realized, $\tilde x^N_e$ is $\min\{f_e/p, 2 \tau/p \}$ at step $1$ of the procedure. Let $\tilde x^N_v =\sum_{e: e \in (S_p \cap N), v \in e} \tilde x^N_e$.
Without looking at the realization of other edges, let $e_1, e_2, \cdots, e_k$ be the non-crucial edges in $S \cap N$ incident to $v$ except the edge $e$. For each edge $e_i$, let $X_i$ be a random variable which is $0$ if $e_i$ is not realized and otherwise is $\min\{f_{e_i}/p, 2\tau/p\}$. Then, for each edge $e_i$, we have
$$
\E[X_i] = p \cdot \min\{f_{e_i}/p, 2\tau /p\} = \min\{f_{e_i}, 2 \tau \} \,.
$$
Let $f^N_v = \sum_{e_i} f_{e_i}$, in the following claim we show that $f^N_v$ is a good approximate of $q^N_v$. Specifically, the claim is as follows.
\begin{claim}
\label{clm:fcloseq}
With probability at least $1-\epsilon$,
$$
\max\{f^N_v, \epsilon\} \le (1+\epsilon) \max\{q^N_v, \epsilon\} \,.
$$
\end{claim}
\begin{proof}
At each round of the algorithm, each edge $e_i$ is sampled with probability $q_i$. Therefore, the probability that vertex $v$ is matched using one of the edges $e_1, e_2, \cdots, e_k$ is at most $\sum_{e_i} q^N_{e_i} \le q^N_v$. Recall that $f_{e_i}$ is the fraction of the matching picked by Algorithm \ref{alg:nonadaptive} that contains $e_i$. Therefore, $f^N_v$ is the fraction of the matchings that vertex $v$ is matched using one of the edges $e_1, e_2, \cdots, e_k$. Therefore, $E[f^N_v] \le q^N_v$. By Hoeffding's inequality we have
$$
P \Big [ f^N_v - q^N_v \ge \epsilon^2 \Big ] \le \exp \big (-2 (R-1) \epsilon^4 \big ) \,.
$$ 
The reason that we have $R-1$ instead of $R$ in the inequality above is that we already know that edge $e$ is realized in one round of the algorithm and we are arguing on other rounds.
Therefore,
$$
P \Big [ f^N_v - q^N_v \ge \epsilon^2 \Big ] \le \exp \big (-2 (R-1) \epsilon^4 \big ) \le \epsilon \,.
$$
Therefore, with probability at least $1-\epsilon$, we have $f^N_v - q^N_v \le \epsilon^2$. It implies that with probability of at least $1-\epsilon$, we have
$$
\max\{f^N_v, \epsilon\} - \max\{q^N_v, \epsilon\}\le \epsilon^2 \,.
$$
Thus, with probability at least $1-\epsilon$,
$$
\max\{f^N_v, \epsilon\} \le \max\{q^N_v, \epsilon\} + \epsilon^2 \le (1+\epsilon) \max\{q^N_v, \epsilon\}.
$$
\end{proof}
For each edge $e_i$, we have $E[X_i] \le \min\{f_{e_i}, 2 \tau \} \le f_{e_i}$. Therefore,
$$
\sum_{e_i} \E[X_i] \le f^N_v ,.
$$
At end of step $1$ of Procedure \ref{proc:non-crucial}, $\tilde x^N_v$ is the sum of the $\tilde x^N_e$ for non-crucial edges in $S$ which are incident to $v$. It follows that
$$
\E[\tilde x^N_v] = \sum_{e_i} \E[X_i] + \tilde x^N_e \le \sum_{e_i} \E[X_i]+ 2\tau/p \,.
$$
If $\tilde x^N_v$ is more than the non-crucial budget of the vertex $v$ which is $\max\{q^N_v, \epsilon\}$, in steps $2$ and $3$ of Procedure \ref{proc:non-crucial}, we scale down the fractional matching such that no vertex violates its non-crucial budget. In the rest of the proof we show that the probability that vertex $v$ violates its budget by a large margin is very small. By the Claim \ref{clm:fcloseq}, we know that $f^N_v$ is very close to non-crucial budget of vertex $v$ and we use $f^N_v$ as a approximation of the budget of vertex $v$. More precisely, we show that with probability at least $1-\epsilon$, $\tilde x^N_v - \tilde x^N_e \le ( 1+ \epsilon) f^N_v$. We use $X$ to denote $\tilde x^N_v - \tilde x^N_e$. Let $\mu = \E[X]$, then $\mu = \E[X]= \sum_{i=1}^{k} \E[X_i]$. We use the variant of Chernoff bound that is given in Lemma~\ref{lem:chernoff}. Note that for each random variable $X_i$, we have $X_i \le \min\{f_{e_i}/p,2 \tau/p \}$ and $\E[X_i] = \min\{f_{e_i},2 \tau\}$. Therefore, $X_i \le \E[X_i] / p$. 
We consider two different cases on $\mu$. The first one is when $\mu \le \epsilon/2$, then $2 \mu \le \max\{f^N_v, \epsilon\}$, and we have
\begin{align*}
P \Big[ X > \max\{f^N_v, \epsilon\} \Big] &= P\Big[X - \max\{f^N_v, \epsilon\}/2 \ge \max\{f^N_v, \epsilon\}/2\Big] \\
&\le P \Big[ X - \mu \ge \max\{f^N_v, \epsilon\}/2 \Big]
\\ &\le \exp \Big(-\dfrac{ \max\{f^N_v, \epsilon\}/2 }{6 \tau/p} \Big) & \text{By Chernoff bound.}
\\ &\le \exp \Big(-\dfrac{\epsilon}{12 \tau/p}\Big)
\\ &\le \exp(-1/\epsilon^2)
\\ &\le \epsilon
\end{align*}
The remaining case is when $\mu > \epsilon/2$. In this case we have
\begin{align*}
P [ X > (1+\epsilon) \max\{f^N_v, \epsilon\} ] & \le P [ X > (1+\epsilon) \mu  ] \\
&\le \exp \Big(-\dfrac{\epsilon^2 \mu}{ 2 \tau /p }\Big) & \text{By Chernoff bound.}
\\ &\le \exp \Big(-\dfrac{\epsilon^3}{6 \tau / p }\Big)  & \text{Since } \mu > \epsilon/2
\\ &\le \exp \Big(-\dfrac{1}{\log (1/\epsilon)}\Big) \,.
\\ &= \epsilon \,.
\end{align*}
Which proves the last case. Therefore with probability at least $1-\epsilon$, we have $X \le (1+\epsilon)  \max\{f^N_v, \epsilon\}$. And, with probability at least $1-\epsilon$ we have
$$
\tilde x^N_v \le (1+\epsilon)  \max\{f^N_v, \epsilon\} + \tilde x^N_e \le (1+\epsilon)  \max\{f^N_v, \epsilon\} + \epsilon \le (1+2\epsilon)  \max\{f^N_v, \epsilon\} \,.
$$
Combining with Claim \ref{clm:fcloseq}, with probability at least $1-2\epsilon$ we have
$$
\tilde x^N_v  \le (1+2\epsilon) \max\{f^N_v, \epsilon\} \le (1+2\epsilon) (1+\epsilon) \max\{q^N_v, \epsilon\} \le (1+5\epsilon) \max\{q^N_v, \epsilon\}  \,.
$$
Therefore, with probability at least $1-2\epsilon$ we have
$$\max\{q^N_v, \epsilon\}/\tilde x^N_v \ge \dfrac{1}{1+5\epsilon} \ge 1-5\epsilon \,,$$
which proves the claim.
\end{proof}

By Claim \ref{clm:non-crucialedgeschernoff}, the probability that for an edge $e$, $\tilde x^N_e$ multiplied by a factor less than $1-5\epsilon$ by one of end points is at most $2\epsilon$. Therefore, by union bound, the probability that none of its end points multiply $\tilde x^N_e$ by a factor less than $1-5\epsilon$ is at least $1-4\epsilon$, i.e., with probability at least $1-4\epsilon$, $s_e \ge 1-5\epsilon$. Therefore,
\begin{align*}
\E \big[ \sum_{e \in S_p \cap N} x^N_e . w_e \big] &=  \sum_{e \in S_p \cap N} \E[x^N_e] . w_e \\
&= \sum_{e \in S_p \cap N} \E[\tilde x^N_e . s_e] . w_e \\
&\ge \sum_{e \in S_p \cap N} (1-4\epsilon)(1-5\epsilon) \E[\tilde x^N_e] . w_e &\text{By Claim \ref{clm:non-crucialedgeschernoff}.} \\
&\ge (1-9\epsilon) \sum_{e \in S_p \cap N} \E[\tilde x^N_e] . w_e \\
&\ge (1-9 \epsilon) (1-\epsilon) \qw(N) & \text{By Claim \ref{clm:non-crucial}.} \\
&\ge (1-10 \epsilon) \qw(N) \,.
\end{align*}
\end{proof}

\subsection{Other Omitted Proofs}\label{sec:otheromitted}
\begin{proof}[Proof of Lemma~\ref{lem:sampleallcrucial}]
	Let $e \in C$ be a crucial edge. We show that Algorithm~\ref{alg:nonadaptive} samples $e$ with probability at least $1-\epsilon$. Let $p^\star_e$ be the probability that Algorithm~\ref{alg:nonadaptive} samples $e$. By Observation \ref{obs:samplingprob}, we have
	$
		1-p^\star_e = (1-q_e)^R
	$. Since $e$ is crucial, we have $q_e \ge \tau$. Thus,
	$ 1-p^\star_e \le (1-\tau)^R $. Note that $R > \frac{\log(1/\epsilon)}{\tau}$. Therefore, $1-p^\star_e$ is at most
	\begin{align}
	\label{ieq:ps}
	1-p^\star_e \le (1-\tau)^\frac{\log(1/\epsilon)}{\tau} = ((1-\tau)^{(1/\tau)})^{\log(1/\epsilon)} \,.
	\end{align}
	We can use the fact that for $(1-x)^{1/x} \leq 1/e$ (see Lemma~\ref{lem:oneovere}) to simplify this bound. Combined with inequality (\ref{ieq:ps}), we have
	$$
	1-p^\star_e \le \Big((1-\tau)^{(1/\tau)}\Big)^{\log(1/\epsilon)} \le \Big(\frac{1}{e}\Big)^{\log(1/\epsilon)} = \dfrac{1}{e^{\log(1/\epsilon)}} = \epsilon \,.
	$$
	Therefore, we have $p^\star_e \ge 1-\epsilon$. Now that we know that each crucial edge is in $S$ with probability at least $1-\epsilon$, we can prove the lemma as follows:
	\begin{align*}
		E[\qw(S \cap C)] &= \E\Big[ \sum_{e \in C} \qw_e \cdot \mathbbm{1}_{S}(e)\Big] &  \text{($\mathbbm{1}_{S}(e) = 1$ if $e \in S$ and 0 otherwise.)} \\
		& =  \sum_{e \in C} \E[\qw_e \cdot \mathbbm{1}_{S}(e)] & \text{By linearity of expectation.} \\
		& = \sum_{e \in C} \qw_e \cdot \E[\mathbbm{1}_{S}(e)] \\
		& = \sum_{e \in C} \qw_e \cdot p^\star_e \\
		& \ge \sum_{e \in C} \qw_e \cdot (1-\epsilon)\\
		& = (1-\epsilon) \sum_{e \in C} \qw_e \\
		& = (1-\epsilon) \qw(C) \,. 
	\end{align*}
	as desired.
\end{proof}

\begin{proof}[Proof of Lemma~\ref{lem:mathratio}]
	We use induction on the value of $n$. 
	
	\paragraph{Base case.} Suppose for the base case that $n=1$.\footnote{To help sanity check the rather technical proof of the base case, we also refer the reader to \href{http://www.wolframalpha.com/input/?i=Maximize\%5B\%7B(h*l)\%2F(h+\%2B+0.5+l),+0+\%3C\%3D+h+\%3C\%3D+1,+0+\%3C\%3D+l+\%3C\%3D+1,++++h+\%2B+l+\%3C\%3D+1\%7D,+\%7Bh,+l\%7D\%5D}{this link on wolframalpha.com.}} We need to prove that for any $a, b \geq 0$ with $0 < a + b \leq 1$ we have $f(a, b) = \frac{ab}{a+b/2} \leq 6-4\sqrt{2}$. For this, we first argue that $f(a, b)$ is maximized when $a+b = 1$. To do this, we show that $f(1-b, b) - f(a, b) \geq 0$ for any $a$ and $b$ that satisfy the conditions above. If $b=0$ or $a=0$ both $f(1-b, b)$ and $f(a, b)$ will be zero and the equation is trivially true, thus assume $a \not= 0$ and $b\not=0$. We have
	\begin{align}
		\nonumber f(1-b, b) - f(a, b) &= \frac{(1-b)b}{1-b+b/2} - \frac{ab}{a+b/2} \\
		\nonumber &= \frac{(1-b)b}{1-b/2} - \frac{ab}{a+b/2}\\
		\nonumber &= \frac{(1-b)b (a+b/2) - ab(1-b/2)}{(1-b/2)(a+b/2)}\\
		\nonumber &= \frac{(ab-ab^2+b^2/2-b^3/2) - (ab - ab^2/2)}{a + b/2 - ab/2 - b^2/4}\\
		\nonumber &= \frac{ab-ab^2+b^2/2-b^3/2 -ab + ab^2/2}{a + b/2 - ab/2 - b^2/4}\\
		\nonumber &= \frac{-ab^2/2+b^2/2-b^3/2}{a + b/2 - ab/2 - b^2/4}\\
		&= \frac{-ab^2+b^2-b^3}{2a + b - ab - b^2/2}.\label{eq:BXKJ}
	\end{align}
	It suffices to show that both the numerator and the denominator of the fraction above are non-negative to show that $f(1-b, b) - f(a, b) \geq 0$. For the numerator, we should show that $-ab^2 + b^2 - b^3 > 0$ or equivalently $b^2 > ab^2 + b^3$. Due to our assumption of $0 < b$, we can divide both sides by $b^2$ to get $1 > a + b$ which is always true as it is part of our initial assumptions on the values of $a$ and $b$. For the denominator, we have to show $2a+b-ab-b^2/2 \geq 0$ or equivalently $2a+b \geq ab + b^2/2$. We have $2a > ab$ since $a, b \in (0, 1)$ and we have $b \geq b^2/2$ since $b \in (0, 1)$. Summing up the two inequalities we get our desired bound that $2a + b \geq ab + b^2/2$; concluding the claim that the fraction in (\ref{eq:BXKJ}) is non-negative and that $f(1-b, b) - f(a, b) \geq 0$.
	
	By the discussion above, to prove the base case, it suffices to find the minimum value of $g(b) := f(1-b, b)$ for $b \in (0, 1)$. Taking the derivative of $g$, we have $g'(b) = \frac{2b^2-8b+4}{(2-b)^2}.$ Setting this equal to zero to get the critical points, we get two solutions of $2-\sqrt{2}$ and $2+\sqrt{2}$. The latter is out of the $(0, 1)$ range and thus the only relevant critical point is when $b = 2-\sqrt{2}$. Therefore we have $$\max_{0\leq b \leq 1} g(b) = g(2-\sqrt{2}) = \frac{\big(1-(2-\sqrt{2})\big)(2-\sqrt{2})}{\big(1-(2-\sqrt{2})\big)+ (2-\sqrt{2})/2} = \frac{3 \sqrt{2}-4}{1/\sqrt{2}} = 6-4 \sqrt{2},$$ implying that for any $a, b \geq 0$ with $0 \leq a+b \leq 1$, we have $f(a, b) \leq g(b) \leq 6-4\sqrt{2}$ as desired for the base case.
	
	\paragraph{Induction step.} Fix numbers $a_1, \ldots, a_{n+1}$ and $b_1, \ldots, b_{n+1}$ that satisfy the conditions of the lemma. Suppose, as induction hypothesis, that we have
	\begin{equation}\label{eq:GPMU}
	\frac{\sum_{i=1}^{n}a_ib_i}{\sum_{i=1}^{n}a_i+\frac{b_i}{2}} \leq 6 - 4\sqrt{2}	
	\end{equation}
	or equivalently,
	\begin{equation}\label{eq:MEXC}
		 \sum_{i=1}^{n}a_ib_i \leq (6 - 4\sqrt{2})\Big(\sum_{i=1}^{n}a_i+\frac{b_i}{2}\Big).
	\end{equation}
	Our goal is to show that 
	\begin{equation}\label{eq:YCXJ}
		\frac{\sum_{i=1}^{n+1}a_ib_i}{\sum_{i=1}^{n+1}a_i+\frac{b_i}{2}} = \frac{\Big(\sum_{i=1}^{n}a_ib_i\Big) + a_{n+1}b_{n+1}}{\Big( \sum_{i=1}^{n}a_i+\frac{b_i}{2}\Big) + a_{n+1}+\frac{b_{n+1}}{2}} \overset{?}{\leq} 6 - 4\sqrt{2}.	
	\end{equation}
	Note that if either of $a_{n+1}$ or $b_{n+1}$ equals 0, then the inequality above is trivially true since the numerator would be equal to that of (\ref{eq:GPMU}) while the denominator is no less than that of (\ref{eq:GPMU}). Thus assume that both $a_{n+1}$ and $b_{n+1}$ are positive. As shown for the base case, we have
	\begin{equation*}
		\frac{a_{n+1}b_{n+1}}{a_{n+1}+\frac{b_{n+1}}{2}} \leq 6-4\sqrt{2},
	\end{equation*}
	which means,
	\begin{equation}\label{eq:ABXJ}
		 a_{n+1}b_{n+1} \leq (6-4\sqrt{2})\big( a_{n+1}+\frac{b_{n+1}}{2}\big).
	\end{equation}
	Replacing (\ref{eq:ABXJ}) and (\ref{eq:MEXC}) into the left-side of the inequality in (\ref{eq:YCXJ}) we get
	\begin{align*}
	\frac{\Big(\sum_{i=1}^{n}a_ib_i\Big) + a_{n+1}b_{n+1}}{\Big( \sum_{i=1}^{n}a_i+\frac{b_i}{2}\Big) + a_{n+1}+\frac{b_{n+1}}{2}} &\leq \frac{\bigg( (6 - 4\sqrt{2})\Big(\sum_{i=1}^{n}a_i+\frac{b_i}{2}\Big) \bigg) + (6-4\sqrt{2})\big( a_{n+1}+\frac{b_{n+1}}{2}\big)}{\Big( \sum_{i=1}^{n}a_i+\frac{b_i}{2}\Big) + a_{n+1}+\frac{b_{n+1}}{2}}	\\
	& \leq \frac{(6-4\sqrt{2})\bigg( \Big( \sum_{i=1}^{n}a_i+\frac{b_i}{2}\Big) + a_{n+1}+\frac{b_{n+1}}{2}\bigg)}{\Big( \sum_{i=1}^{n}a_i+\frac{b_i}{2}\Big) + a_{n+1}+\frac{b_{n+1}}{2}}\\
	& \leq 6-4\sqrt{2},
	\end{align*}
	which is the desired bound of inequality (\ref{eq:YCXJ}).
\end{proof}

%% file: usefulbounds.tex
\section{Appendix: Used Inequalities}

\begin{lemma}[Chernoff bound]\label{lem:chernoff} 
Given a real number $b>0$, let $X_1, X_2, \ldots X_n$ be $n$ random variables such that $0 \le X_i \le b$ for every $X_i$. Let $X= \sum_{i=1}^n X_i$ and $\mu= \E[X]$. Then for $0 \le \delta \le 1$,
$$
P [ X \ge (1+ \delta)\mu ] \le \exp\Big(-\dfrac{\delta^2 \mu}{3 b}\Big) \,.
$$
Also, for $\delta \ge 1$,
$$
P [ X \ge (1+ \delta)\mu ] \le \exp\Big(-\dfrac{\delta \mu}{3 b}\Big) \,.
$$
\end{lemma}

\begin{lemma} \label{lem:oneovere}
 Let $f(x) = (1-x)^{1/x}$. Then, for any $0<x \le 1$, $f(x) \le \frac{1}{e}$. 
\end{lemma}
	
\begin{proof} [Proof of Lemma~\ref{lem:oneovere}]
We want to find the maximum value of $f$ for $0 < x \le 1$. By taking derivative with respect to $x$ we have
$$f'(x) = -\frac{1}{x} \,.$$
Therefore, $f$ is an decreasing function in $x$ and its maximum value is when $x$ is very close to $0$. Formally,
$$
f(x) \le \lim_{x \to 0} f(x)
$$
It can easily be verified that limit of $f$ as $x$ approaches $0$ is $1/e$. Therefore, 
$$f(x) \le \lim_{x \to 0} f(x) \le \frac{1}{e}.$$
\end{proof}

%% file: ms.bbl
\newcommand{\etalchar}[1]{$^{#1}$}
\begin{thebibliography}{AAGK15}

\bibitem[AAGK15]{DBLP:conf/soda/AndersonAGK15}
Ross Anderson, Itai Ashlagi, David Gamarnik, and Yash Kanoria.
\newblock {A Dynamic Model of Barter Exchange}.
\newblock In {\em Proceedings of the Twenty-Sixth Annual {ACM-SIAM} Symposium
  on Discrete Algorithms, {SODA} 2015, San Diego, CA, USA, January 4-6, 2015},
  pages 1925--1933, 2015.

\bibitem[AAGR15]{anderson2015finding}
Ross Anderson, Itai Ashlagi, David Gamarnik, and Alvin~E Roth.
\newblock {Finding Long Chains in Kidney Exchange Using the Traveling Salesman
  Problem}.
\newblock {\em Proceedings of the National Academy of Sciences},
  112(3):663--668, 2015.

\bibitem[AB]{assadisosa}
Sepehr Assadi and Aaron Bernstein.
\newblock Towards a unified theory of sparsification for matching problems.
\newblock {\em 2nd Symposium on Simplicity in Algorithms, {SOSA} 2019, to
  appear.}

\bibitem[Ada11]{DBLP:journals/ipl/Adamczyk11}
Marek Adamczyk.
\newblock {Improved analysis of the greedy algorithm for stochastic matching}.
\newblock {\em Inf. Process. Lett.}, 111(15):731--737, 2011.

\bibitem[AKL16]{DBLP:conf/sigecom/AssadiKL16}
Sepehr Assadi, Sanjeev Khanna, and Yang Li.
\newblock {The Stochastic Matching Problem with (Very) Few Queries}.
\newblock In {\em Proceedings of the 2016 {ACM} Conference on Economics and
  Computation, {EC} '16, Maastricht, The Netherlands, July 24-28, 2016}, pages
  43--60, 2016.

\bibitem[AKL17]{DBLP:conf/sigecom/AssadiKL17}
Sepehr Assadi, Sanjeev Khanna, and Yang Li.
\newblock {The Stochastic Matching Problem: Beating Half with a Non-Adaptive
  Algorithm}.
\newblock In {\em Proceedings of the 2017 {ACM} Conference on Economics and
  Computation, {EC} '17, Cambridge, MA, USA, June 26-30, 2017}, pages 99--116,
  2017.

\bibitem[ALG14]{DBLP:conf/sigecom/AkbarpourLG14}
Mohammad Akbarpour, Shengwu Li, and Shayan~Oveis Gharan.
\newblock {Dynamic Matching Market Design}.
\newblock In {\em {ACM} Conference on Economics and Computation, {EC} '14,
  Stanford , CA, USA, June 8-12, 2014}, page 355, 2014.

\bibitem[AS09]{DBLP:conf/ijcai/AwasthiS09}
Pranjal Awasthi and Tuomas Sandholm.
\newblock {Online Stochastic Optimization in the Large: Application to Kidney
  Exchange}.
\newblock In {\em {IJCAI} 2009, Proceedings of the 21st International Joint
  Conference on Artificial Intelligence, Pasadena, California, USA, July 11-17,
  2009}, pages 405--411, 2009.

\bibitem[BDH{\etalchar{+}}14]{DBLP:journals/corr/BlumHPS14}
Avrim Blum, John~P. Dickerson, Nika Haghtalab, Ariel~D. Procaccia, Tuomas
  Sandholm, and Ankit Sharma.
\newblock {Ignorance is Almost Bliss: Near-Optimal Stochastic Matching With Few
  Queries}.
\newblock {\em CoRR}, abs/1407.4094, 2014.

\bibitem[BDH{\etalchar{+}}15]{DBLP:conf/sigecom/BlumDHPSS15}
Avrim Blum, John~P. Dickerson, Nika Haghtalab, Ariel~D. Procaccia, Tuomas
  Sandholm, and Ankit Sharma.
\newblock {Ignorance is Almost Bliss: Near-Optimal Stochastic Matching With Few
  Queries}.
\newblock In {\em Proceedings of the Sixteenth {ACM} Conference on Economics
  and Computation, {EC} '15, Portland, OR, USA, June 15-19, 2015}, pages
  325--342, 2015.

\bibitem[BGL{\etalchar{+}}12]{DBLP:journals/algorithmica/BansalGLMNR12}
Nikhil Bansal, Anupam Gupta, Jian Li, Juli{\'{a}}n Mestre, Viswanath Nagarajan,
  and Atri Rudra.
\newblock {When {LP} Is the Cure for Your Matching Woes: Improved Bounds for
  Stochastic Matchings}.
\newblock {\em Algorithmica}, 63(4):733--762, 2012.

\bibitem[BGPS13]{DBLP:conf/sigecom/BlumGPS13}
Avrim Blum, Anupam Gupta, Ariel~D. Procaccia, and Ankit Sharma.
\newblock {Harnessing the power of two crossmatches}.
\newblock In {\em {ACM} Conference on Electronic Commerce, {EC} '13,
  Philadelphia, PA, USA, June 16-20, 2013}, pages 123--140, 2013.

\bibitem[BR18]{DBLP:conf/sigecom/BehnezhadR18}
Soheil Behnezhad and Nima Reyhani.
\newblock {Almost Optimal Stochastic Weighted Matching with Few Queries}.
\newblock In {\em Proceedings of the 2018 {ACM} Conference on Economics and
  Computation, Ithaca, NY, USA, June 18-22, 2018}, pages 235--249, 2018.

\bibitem[CIK{\etalchar{+}}09]{DBLP:conf/icalp/ChenIKMR09}
Ning Chen, Nicole Immorlica, Anna~R. Karlin, Mohammad Mahdian, and Atri Rudra.
\newblock {Approximating Matches Made in Heaven}.
\newblock In {\em Automata, Languages and Programming, 36th International
  Colloquium, {ICALP} 2009, Rhodes, Greece, July 5-12, 2009, Proceedings, Part
  {I}}, pages 266--278, 2009.

\bibitem[CTT12]{DBLP:conf/icalp/CostelloTT12}
Kevin~P. Costello, Prasad Tetali, and Pushkar Tripathi.
\newblock {Stochastic Matching with Commitment}.
\newblock In {\em Automata, Languages, and Programming - 39th International
  Colloquium, {ICALP} 2012, Warwick, UK, July 9-13, 2012, Proceedings, Part
  {I}}, pages 822--833, 2012.

\bibitem[DPS12]{DBLP:conf/aaai/DickersonPS12}
John~P. Dickerson, Ariel~D. Procaccia, and Tuomas Sandholm.
\newblock {Dynamic Matching via Weighted Myopia with Application to Kidney
  Exchange}.
\newblock In {\em Proceedings of the Twenty-Sixth {AAAI} Conference on
  Artificial Intelligence, July 22-26, 2012, Toronto, Ontario, Canada.}, 2012.

\bibitem[DPS13]{DBLP:conf/sigecom/DickersonPS13}
John~P. Dickerson, Ariel~D. Procaccia, and Tuomas Sandholm.
\newblock {Failure-aware kidney exchange}.
\newblock In {\em {ACM} Conference on Electronic Commerce, {EC} '13,
  Philadelphia, PA, USA, June 16-20, 2013}, pages 323--340, 2013.

\bibitem[DS15]{DBLP:conf/aaai/DickersonS15}
John~P. Dickerson and Tuomas Sandholm.
\newblock {FutureMatch: Combining Human Value Judgments and Machine Learning to
  Match in Dynamic Environments}.
\newblock In {\em Proceedings of the Twenty-Ninth {AAAI} Conference on
  Artificial Intelligence, January 25-30, 2015, Austin, Texas, {USA.}}, pages
  622--628, 2015.

\bibitem[Edm65]{edmonds1965maximum}
Jack Edmonds.
\newblock {Maximum matching and a polyhedron with 0, 1-vertices}.
\newblock {\em Journal of research of the National Bureau of Standards B},
  69(125-130):55--56, 1965.

\bibitem[GN13]{DBLP:conf/ipco/GuptaN13}
Anupam Gupta and Viswanath Nagarajan.
\newblock {A Stochastic Probing Problem with Applications}.
\newblock In {\em Integer Programming and Combinatorial Optimization - 16th
  International Conference, {IPCO} 2013, Valpara{\'{\i}}so, Chile, March 18-20,
  2013. Proceedings}, pages 205--216, 2013.

\bibitem[MO14]{DBLP:journals/jea/ManloveO14}
David~F. Manlove and Gregg O'Malley.
\newblock {Paired and Altruistic Kidney Donation in the {UK:} Algorithms and
  Experimentation}.
\newblock {\em {ACM} Journal of Experimental Algorithmics}, 19(1), 2014.

\bibitem[Sch03]{schrijver2003combinatorial}
Alexander Schrijver.
\newblock {\em {Combinatorial Optimization: Polyhedra and Efficiency}},
  volume~24.
\newblock Springer Science \& Business Media, 2003.

\bibitem[{\"U}nv10]{unver2010dynamic}
M~Utku {\"U}nver.
\newblock {Dynamic Kidney Exchange}.
\newblock {\em The Review of Economic Studies}, 77(1):372--414, 2010.

\bibitem[YM18]{DBLP:conf/soda/YamaguchiM18}
Yutaro Yamaguchi and Takanori Maehara.
\newblock {Stochastic Packing Integer Programs with Few Queries}.
\newblock In {\em Proceedings of the Twenty-Ninth Annual {ACM-SIAM} Symposium
  on Discrete Algorithms, {SODA} 2018, New Orleans, LA, USA, January 7-10,
  2018}, pages 293--310, 2018.

\end{thebibliography}
